\documentclass[peerreview,a4paper,12pt]{IEEEtran}

\usepackage{amsthm}
\usepackage{amsmath}
\usepackage{amsfonts,amssymb,amsmath}
\usepackage{graphicx}
\usepackage{epsfig}
\usepackage[applemac]{inputenc}
\usepackage{latexsym}

\newtheorem{mydef}{Definition}
\newtheorem{mycor}{Corollary}
\newtheorem{myprop}{Proposition}
\newtheorem{mythe}{Theorem}
\newtheorem{mylem}{Lemma}
\newtheorem{mynot}{Notation}
\newtheorem{myrem}{Remark}

\DeclareMathOperator*{\Y}{Y}
\DeclareMathOperator*{\YZ}{YZ}
\DeclareMathOperator*{\X}{\hat{X}}
\DeclareMathOperator*{\XDY}{D}
\DeclareMathOperator*{\DY}{\Delta\!Y}
\DeclareMathOperator*{\XZ}{\hat{X}Z}

\begin{document}

\sloppy

\title{Fourier Analysis of MAC Polarization}

\author{
Rajai Nasser, Emre Telatar
  \thanks{This paper was presented in part at the IEEE International Symposium on Information Theory, Hong Kong, June 2015.}
}



\maketitle

\begin{abstract}
One problem with MAC polar codes that are based on MAC polarization is that they may not achieve the entire capacity region. The reason behind this problem is that MAC polarization sometimes induces a loss in the capacity region. This paper provides a single letter necessary and sufficient condition which characterizes all the MACs that do not lose any part of their capacity region by polarization.
\end{abstract}

\section{Introduction}

Polar coding is a low complexity coding technique invented by Ar{\i}kan that achieves the capacity of symmetric binary input channels \cite{Arikan}. The error probability of polar codes was shown to be $o(2^{-N^{\frac{1}{2}-\epsilon}})$ where $N$ is the block length \cite{ArikanTelatar}. The polar coding construction of Ar{\i}kan transforms a set of identical and independent channels to a set of ``almost perfect" or ``almost useless channels". This phenomenon is called \emph{polarization}.

Polarizing transformations can also be constructed for non-binary input channels. \c{S}a\c{s}o\u{g}lu et al. \cite{SasogluTelAri} generalized Ar{\i}kan's results to channels where the input alphabet size is prime. Park and Barg \cite{ParkBarg} showed that if the size of the input alphabet is of the form $2^r$ with $r>1$, then using the algebraic structure $\mathbb{Z}_{2^r}$ in the polarizing transformation leads to a multilevel polarization phenomenon: while we do not always have polarization to ``almost perfect" or ``almost useless" channels, we always have polarization to channels that are easy to use for communication. Multilevel polarization can be used to construct capacity achieving polar codes.

Sahebi and Pradhan \cite{SahebiPradhan} showed that multilevel polarization also happens if any Abelian group operation on the input alphabet is used. This allows the construction of polar codes for arbitrary discrete memoryless channels (DMC) since any alphabet can be endowed with an Abelian group structure. Polar codes for arbitrary DMCs were also constructed by \c{S}a\c{s}o\u{g}lu \cite{SasS} by using a special quasigroup operation that ensures two-level polarization. The authors showed in \cite{RajTel} that all quasigroup operations are polarizing (in the general multilevel sense) and can be used to construct capacity-achieving polar codes for arbitrary DMCs \cite{RajTelA}.

In the context of multiple access channels (MAC), \c{S}a\c{s}o\u{g}lu et al. showed that if $W$ is a 2-user MAC where the two users have $\mathbb{F}_q$ as input alphabet, then using the addition modulo $q$ for the two users leads to a MAC polarization phenomenon \cite{SasogluTelYeh}. Abbe and Telatar showed that for binary input MACs with $m\geq 2$ users, using the XOR operation for each user is MAC-polarizing \cite{AbbeTelatar}. A problem with the constructions in \cite{SasogluTelYeh} and \cite{AbbeTelatar} is that they do not always achieve the entire capacity region. The reason behind this problem is that MAC polarization sometimes induces a loss in the capacity region.

A characterization of all the polarizing transformations that are based on binary operations --- in both the single-user and the multiple access settings --- can be found in \cite{RajErgI} and \cite{RajErgII}. Abelian group operations are a special case of the characterization in \cite{RajErgII}. Therefore, using Abelian group operations for all users is sufficient to polarize any given MAC.

This paper provides a necessary and sufficient condition that characterizes the set of MACs that do not lose any part of their capacity region by polarization. The characterization that we provide works in the general setting where we have an arbitrary number of users and each user uses an arbitrary Abelian group operation on his input alphabet.

We will show that the reason why a given MAC $W$ loses parts of its capacity region by polarization is because its transition probabilities are not ``aligned", which makes $W$ ``incompatible" with polarization. The ``alignment" condition will be expressed in terms of the Fourier transforms of the transition probabilities of $W$. The use of Fourier analysis in our study should not come as a surprise since the transition probabilities of $W^-$ can be expressed as a convolution of the transition probabilities of $W$. This is what makes Fourier analysis useful for our study because it turns convolutions into multiplications, which are much easier to analyze.

Note that there are alternate polar coding solutions that can achieve the entire capacity region without any loss. These techniques, which are not based on MAC polarization, are hybrid schemes combining single user channel polarization with other techniques. In \cite{SasogluTelYeh}, \c{S}a\c{s}o\u{g}lu et al. used the ``rate splitting/onion peeling" scheme of \cite{RimUr} and \cite{Gretal} to transform any point on the dominant face of an $m$-user MAC into a corner point of a $(2m-1)$-user MAC and then applied single user channel polarization to achieve this corner point. In \cite{Monotone}, Ar{\i}kan used monotone chain rules to construct polar codes for the Slepian-Wolf problem, but the same technique can be used to achieve the entire capacity region of a MAC.

Although the alternate solutions of \cite{SasogluTelYeh} and \cite{Monotone} can achieve the entire capacity region, they are more complicated than MAC polar codes (those that are based on MAC polarization). The alternate solution in \cite{SasogluTelYeh} requires more encoding and decoding complexity because it adds $m-1$ virtual users. Ar{\i}kan's solution \cite{Monotone} does not add significant encoding and decoding complexity, but the code design is much more complicated than that of MAC polar codes. So if we are given a MAC $W$ whose capacity region is preserved by polarization (i.e., MAC polar codes can achieve the entire capacity region of this MAC), then using MAC polar codes for this MAC is preferable to the alternate solutions. One practical implication of this study is that it allows a code designer to determine whether he can use the preferable MAC polar codes to achieve the capacity region.

In section II, we introduce the preliminaries of this paper: we describe the MAC polarization process and explain the discrete Fourier transforms on Abelian groups. In section III, we provide a sufficient condition for the preservation of the capacity region. This sufficient condition, which is relatively easy to understand, gives an intuition that clarifies the necessary and sufficient condition that we prove later. In section IV, we characterize the two-user MACs whose capacity regions are preserved by polarization. Section V generalizes the results of section IV to MACs with arbitrary number of users.

\section{Preliminaries}

A discrete $m$-user multiple access channel (MAC) is an ($m+2$)-tuple $W=(\mathcal{X}_1,\;\mathcal{X}_2,\;\ldots,\;\mathcal{X}_m,\;\mathcal{Z},\;p_W)$ where $\mathcal{X}_1,\;\ldots,\;\mathcal{X}_m$ are finite sets that are called the \emph{input alphabets} of $W$, $\mathcal{Z}$ is a finite set that is called the \emph{output alphabet} of $W$, and $p_W:\mathcal{X}_1\times\mathcal{X}_2\times\ldots\times\mathcal{X}_m \times \mathcal{Z} \rightarrow [0,1]$  is a function satisfying $\forall(x_1,x_2,\ldots,x_m)\in\mathcal{X}_1\times\mathcal{X}_2\times\ldots\times\mathcal{X}_m,\;\displaystyle\sum_{z\in\mathcal{Z}} p_W(x_1,x_2,\ldots,x_m,z)=1$.

We write $W: \mathcal{X}_1\times\mathcal{X}_2\times\ldots\times\mathcal{X}_m \longrightarrow \mathcal{Z}$ to denote that $W$ has $m$ users, $\mathcal{X}_1,\;\mathcal{X}_2,\;\ldots,\;\mathcal{X}_m$ as input alphabets, and $\mathcal{Z}$ as output alphabet. We denote $p_W(x_1,x_2,\ldots,x_m,z)$ by $W(z|x_1,x_2,\ldots,x_m)$ which is interpreted as the conditional probability of receiving $z$ at the output, given that $(x_1,x_2,\ldots,x_m)$ is the input. Note that we use the long arrow ($\longrightarrow$) in the notation $W: \mathcal{X}_1\times\mathcal{X}_2\times\ldots\times\mathcal{X}_m \longrightarrow \mathcal{Z}$ and not the short arrow ($\rightarrow$) which we only use to describe mappings. For example, $W:\mathcal{X}_1\times\mathcal{X}_2\longrightarrow\mathcal{Z}$ denotes a 2-user MAC, while $V:\mathcal{X}_1\times\mathcal{X}_2\to\mathcal{Z}$ denotes a mapping from $\mathcal{X}_1\times\mathcal{X}_2$ to $\mathcal{Z}$.

\subsection{Polarization}

Throughout this paper, $G_1,\ldots,G_m$ are finite Abelian groups. We will use the addition symbol $+$ to denote the group operations of $G_1,\ldots,G_m$. Since every finite Abelian group is isomorphic to the product of cyclic groups, we may assume without loss of generality that $G_1,\ldots,G_m$ are products of cyclic groups. In other words, for every $1\leq i\leq m$, there exist $k_i$ integers $N_{i,1},\ldots,N_{i,k_i}>0$ such that $G_i=\mathbb{Z}_{N_{i,1}}\times\ldots\times\mathbb{Z}_{N_{i,k_i}}$.

\begin{mynot}
Let $W:G_1\times\ldots\times G_m\longrightarrow \mathcal{Z}$ be an $m$-user MAC. We write $(X_1,\ldots,X_m)\stackrel{W}{\longrightarrow}Z$ to denote the following:
\begin{itemize}
\item $X_1,\ldots,X_m$ are independent random variables uniformly distributed in $G_1,\ldots, G_m$ respectively.
\item $Z$ is the output of the MAC $W$ when $X_1,\ldots,X_m$ are the inputs.
\end{itemize}
\end{mynot}

\begin{mynot}
Fix $S\subset \{1,\ldots,m\}$ and let $S=\{i_1,\ldots,i_{|S|}\}$. Define $G_S$ as $$G_S:=\prod_{i\in S} G_i=G_{i_1}\times\ldots\times G_{i_{|S|}}.$$

For every $(x_1,\ldots,x_m)\in G_1\times\ldots\times G_m$, we write $x_S$ to denote $(x_{i_1},\ldots,x_{i_{|S|}})$.
\label{notNotation}
\end{mynot}

\begin{mynot}
Let $W:G_1\times\ldots\times G_m\longrightarrow \mathcal{Z}$ and $(X_1,\ldots,X_m)\stackrel{W}{\longrightarrow}Z$. For every $S\subset \{1,\ldots,m\}$, we write $I_S(W)$ to denote $I(X_S;ZX_{S^c})$. If $S=\{i\}$, we denote $I_{\{i\}}(W)$ as $I_i(W)$.

$I(W):=I_{\{1,\ldots,m\}}(W)=I(X_1,\ldots,X_m;Z)$ is called \emph{the symmetric sum-capacity} of $W$.

The \emph{symmetric capacity region} of an $m$-user MAC $W:G_1\times\ldots\times G_m\longrightarrow \mathcal{Z}$ is defined as:
\begin{align*}
\mathcal{J}(W)=\Big\{(R_1,\ldots&,R_m)\in\mathbb{R}^m:\;\forall S\subset\{1,\ldots,m\},\;0\leq\sum_{i\in S} R_i\leq I_S(W)\Big\}.
\end{align*}
\end{mynot}
Note that $I(W)$ is called the \emph{symmetric} sum-capacity because it is computed using uniform input distributions. The same is true for $\mathcal{J}(W)$.

\begin{mynot}
$\displaystyle\{-,+\}^{\ast}:=\bigcup_{n\geq 0}\{-,+\}^n$, where $\{-,+\}^0=\{\o\}$.
\end{mynot}

\begin{mydef}
Let $W:G_1\times\ldots\times G_m\longrightarrow \mathcal{Z}$. We define the $m$-user MACs $W^-:G_1\times\ldots\times G_m\longrightarrow\mathcal{Z}^2$ and $W^+:G_1\times\ldots\times G_m\longrightarrow\mathcal{Z}^2\times G_1\times\ldots\times G_m$ as follows:
\begin{align*}
W^-(z_1,z_2|u_{11},\ldots,u_{1m})=\sum_{\substack{u_{21}\in G_1\\ \vdots\\u_{2m}\in G_m}} \frac{1}{|G_1|\cdots|G_m|}W(z_1|u_{11}+u_{21},\ldots, u_{1m}+&u_{2m})\\
&\times W(z_2|u_{21},\ldots, u_{2m}),
\end{align*}
and
\begin{align*}
W^+(z_1,z_2,u_{11},\ldots,u_{1m}|u_{21},\ldots,u_{2m})=\frac{1}{|G_1|\cdots|G_m|}W(z_1|u_{11}+u_{21},\ldots&, u_{1m}+u_{2m})\\
&\times W(z_2|u_{21},\ldots, u_{2m}).
\end{align*}
For every $s\in\{-,+\}^{\ast}$, we define the MAC $W^s$ as follows:
\begin{align*}
W^s:=\begin{cases} W &\text{if}\;s=\o,\\ (\ldots((W^{s_1})^{s_2})\ldots)^{s_n}&\text{if}\;s=(s_1,\ldots,s_n).\end{cases}
\end{align*}
\end{mydef}

\vspace*{5mm}

The following remark explains why polarization may induce a loss in the capacity region.

\begin{myrem}
\label{MainRem}
Let $U_1^m$ and $\tilde{U}_1^m$ be two independent random variables uniformly distributed in $G_1\times\ldots\times G_m$. Let $X_1^m=U_1^m+\tilde{U}_1^m$ and $\tilde{X}_1^m=\tilde{U}_1^m$. Let $(X_1,\ldots,X_m)\stackrel{W}{\longrightarrow}Z$ and $(\tilde{X}_1,\ldots,\tilde{X}_m)\stackrel{W}{\longrightarrow}\tilde{Z}$. We have:
\begin{itemize}
\item $I(W)=I(X_1^m;Z)=I(\tilde{X}_1^m;\tilde{Z})$.
\item $I(W^-)=I(U_1^m;Z\tilde{Z})$ and $I(W^+)=I(\tilde{U}_1^m;Z\tilde{Z}U_1^m)$.
\end{itemize}
Hence,
\begin{align*}
2I(W)&=I(X_1^m;Z)+I(\tilde{X}_1^m;\tilde{Z})=I(X_1^m\tilde{X}_1^m;Z\tilde{Z})=I(U_1^m\tilde{U}_1^m;Z\tilde{Z})\\ &=I(U_1^m;Z\tilde{Z})+I(\tilde{U}_1^m;Z\tilde{Z}|U_1^m)\stackrel{(a)}{=}I(U_1^m;Z\tilde{Z})+I(\tilde{U}_1^m;Z\tilde{Z}U_1^m)=I(W^-)+I(W^+),
\end{align*}
where (a) follows from the fact that $U_1^m$ is independent of $\tilde{U}_1^m$.

Therefore, the symmetric sum-capacity is preserved by polarization. On the other hand, $I_S$ might not be preserved if $S\subsetneq\{1,\ldots,m\}$.

For example, consider the two-user MAC case. Let $W:G_1\times G_2\longrightarrow\mathcal{Z}$. Let $(U_1,V_1)$ and $(U_2,V_2)$ be two independent random pairs uniformly distributed in $G_1\times G_2$. Let $X_1=U_1+U_2$, $X_2=U_2$, $Y_1=V_1+V_2$ and $Y_2=V_2$. Let $(X_1,Y_1)\stackrel{W}{\longrightarrow}Z_1$ and $(X_2,Y_2)\stackrel{W}{\longrightarrow}Z_2$. We have:
\begin{itemize}
\item $I_1(W^-)=I(U_1;Z_1Z_2V_1)$ and $I_1(W^+)=I(U_2;Z_1Z_2U_1V_1V_2)$.
\item $I_2(W^-)=I(V_1;Z_1Z_2U_1)$ and $I_2(W^+)=I(V_2;Z_1Z_2U_1V_1U_2)$.
\end{itemize}

On the other hand, we have:
\begin{itemize}
\item $I_1(W)=I(X_1;Z_1Y_1)=I(X_2;Z_2Y_2)$.
\item $I_2(W)=I(Y_1;Z_1X_1)=I(Y_2;Z_2X_2)$.
\end{itemize}

Therefore,
\begin{align}
2I_1(W)&=I(X_1;Z_1Y_1)+I(X_2;Z_2Y_2)=I(X_1X_2;Z_1Z_2Y_1Y_2) =I(U_1U_2;Z_1Z_2V_1V_2)\nonumber\\
&=I(U_1;Z_1Z_2V_1V_2)+I(U_2;Z_1Z_2V_1V_2U_1)\stackrel{(a)}{\geq} I(U_1;Z_1Z_2V_1)+I(U_2;Z_1Z_2V_1V_2U_1) \label{mainineq} \\
&= I_1(W^-)+I_1(W^+),\nonumber
\end{align}
where (a) follows from the fact that $I(U_1;Z_1Z_2V_1V_2)=I(U_1;Z_1Z_2V_1)+ I(U_1;V_2|Z_1Z_2V_1)\geq I(U_1;Z_1Z_2V_1)$.

Similarly,
\begin{align*}
2I_2(W)&=I(Y_1;Z_1X_1)+I(Y_2;Z_2X_2)=I(Y_1Y_2;Z_1Z_2X_1X_2) =I(V_1V_2;Z_1Z_2U_1U_2)\\
&=I(V_1;Z_1Z_2U_1U_2)+I(V_2;Z_1Z_2U_1U_2V_1)\geq I(V_1;Z_1Z_2U_1)+I(V_2;Z_1Z_2U_1U_2V_1) \\
&= I_2(W^-)+I_2(W^+).
\end{align*}

Note that $\displaystyle\frac{1}{2^0}\sum_{s\in\{-,+\}^0}I_1(W^s)=\frac{1}{1}I_1(W^{\o})=I_1(W)\leq I_1(W)$. Now let $n\geq 0$ and assume that $\displaystyle\frac{1}{2^n}\sum_{s\in\{-,+\}^n}I_1(W^s)\leq I_1(W)$, then
\begin{align*}
\frac{1}{2^{n+1}}\sum_{s\in\{-,+\}^{n+1}}I_1(W^s)&=\frac{1}{2^{n+1}}\sum_{s\in\{-,+\}^{n}}\left(I_1\left(W^{(s,-)}\right)+I_1\left(W^{(s,+)}\right)\right)\stackrel{(a)}{\leq} \frac{1}{2^{n+1}}\sum_{s\in\{-,+\}^{n}}2I_1(W)\\
&=\frac{1}{2^n}\sum_{s\in\{-,+\}^n}I_1(W^s)\leq I_1(W),
\end{align*}
where (a) follows from applying \eqref{mainineq} to $W^s$. We conclude that for every $n\geq 0$ we have:
\begin{equation}
\frac{1}{2^n}\sum_{s\in\{-,+\}^n} I_1(W^s)\leq I_1(W).
\label{eq1}
\end{equation}
Similarly,
\begin{equation}
\frac{1}{2^n}\sum_{s\in\{-,+\}^n} I_2(W^s)\leq I_2(W).
\label{eq2}
\end{equation}

By using a similar induction argument, but using the equality $I\left(W^{(s,-)}\right)+I\left(W^{(s,+)}\right)=2I(W^s)$, we can show that for every $n\geq 0$, we have:
\begin{equation}
\frac{1}{2^n}\sum_{s\in\{-,+\}^n} I(W^s)= I(W).
\label{eq3}
\end{equation}
While \eqref{eq3} shows that polarization preserves the symmetric sum-capacity, \eqref{eq1} and \eqref{eq2} show that polarization may result into a loss in the capacity region.

Similarly, for the $m$-user case, we have
$$\frac{1}{2^n}\sum_{s\in\{-,+\}^n} I_S(W^s)\leq I_S(W),\;\forall S\subsetneq\{1,\ldots,m\}.$$
\end{myrem}

\begin{mydef}
\label{defStarPres}

Let $S\subset\{1,\ldots,m\}$. We say that polarization \emph{$\ast$-preserves} $I_S$ for $W$ if for all $n\geq 0$ we have:
\begin{equation*}
\frac{1}{2^n}\sum_{s\in\{-,+\}^n} I_S(W^s)=I_S(W).
\end{equation*}

If polarization $\ast$-preserves $I_S$ for every $S\subset \{1,\ldots,m\}$, we say that polarization \emph{$\ast$-preserves} the symmetric capacity region for $W$.
\end{mydef}

\begin{myrem}
If polarization $\ast$-preserves the symmetric capacity region for $W$, then the entire symmetric capacity region can be achieved by polar codes.
\end{myrem}

Section IV provides a characterization of two-user MACs whose $I_1$ is $\ast$-preserved by polarization. Section V generalizes the results of section IV and provides a characterization of $m$-user MACs whose $I_S$ is $\ast$-preserved by polarization, where $S\subsetneq\{1,\ldots,m\}$. This yields a complete characterization of the MACs with $\ast$-preserved symmetric capacity regions.

\subsection{Discrete Fourier Transform on finite Abelian Groups}

A tool that we are going to need for the analysis of the polarization process is the discrete Fourier transform (DFT) on finite Abelian groups. Since every finite Abelian group is isomorphic to the product of cyclic groups, the DFT on finite Abelian groups can be defined based on the usual multidimensional DFT.

\begin{mydef}
The $k$-dimensional discrete Fourier transform of a mapping $f:\mathbb{Z}_{N_1}\times\ldots\times\mathbb{Z}_{N_k}\rightarrow\mathbb{C}$ is the mapping $\hat{f}: \mathbb{Z}_{N_1}\times\ldots\times\mathbb{Z}_{N_k}\rightarrow\mathbb{C}$ defined as:
\begin{align*}
\hat{f}(\hat{x}_1,\ldots,\hat{x}_k)=\sum_{x_1\in\mathbb{Z}_{N_1},\ldots,x_k\in\mathbb{Z}_{N_k}}f(x_1,\ldots,x_k)e^{-j\frac{2\pi \hat{x}_1x_1}{N_1}\ldots-j\frac{2\pi \hat{x}_kx_k}{N_k}}.
\end{align*}
\end{mydef}

\begin{mynot}
Let $G=\mathbb{Z}_{N_1}\times\ldots\times\mathbb{Z}_{N_k}$ be a finite Abelian group. For every $x=(x_1,\ldots,x_k)\in G$ and every $\hat{x}=(\hat{x}_1,\ldots,\hat{x}_k)\in G$, define $\langle \hat{x},x\rangle\in\mathbb{R}$ as:
$$\langle \hat{x},x\rangle:=\frac{\hat{x}_1x_1}{N_1}+\ldots+\frac{\hat{x}_kx_k}{N_k}\in\mathbb{R}.$$
Using this notation, the DFT on $G$ has a compact formula:
\begin{align*}
\hat{f}(\hat{x})=\sum_{x\in G}f(x)e^{-j2\pi\langle \hat{x},x\rangle}.
\end{align*}
\end{mynot}

In the rest of this section, we recall well known properties of DFT.

\begin{myprop}
The inverse DFT is given by the following formula:
\begin{align*}
f(x)=\frac{1}{|G|}\sum_{\hat{x}\in G}\hat{f}(\hat{x})e^{j2\pi\langle \hat{x},x\rangle}.
\end{align*}
\end{myprop}

\begin{mydef}
The convolution of two mappings $f:G\rightarrow\mathbb{C}$ and $g:G\rightarrow\mathbb{C}$ is the mapping $f\ast g:G\rightarrow\mathbb{C}$ defined as:
$$(f\ast g)(x)=\sum_{x'\in G}f(x')g(x-x').$$
We will sometimes write $f(x)\ast g(x)$ to denote $(f\ast g)(x)$.
\end{mydef}

\begin{myprop}
Let $f:G\rightarrow\mathbb{C}$ and $g:G\rightarrow\mathbb{C}$ be two mappings. We have:
\begin{itemize}
\item $\widehat{(f\ast g)}(\hat{x})=\hat{f}(\hat{x})\hat{g}(\hat{x})$.
\item $\displaystyle \widehat{(f\cdot g)}(\hat{x})=\frac{1}{|G|}(\hat{f}\ast\hat{g})(\hat{x})$.
\item If $f_a:G\rightarrow \mathbb{C}$ is defined as $f_a(x)=f(x-a)$, then
$\hat{f}_a(\hat{x})=\hat{f}(\hat{x})e^{-j2\pi\langle\hat{x},a\rangle}$.
\item If $\tilde{f}:G\rightarrow \mathbb{C}$ is defined as $\tilde{f}(x)=f(-x)$, then
$\hat{\tilde{f}}(\hat{x})=\hat{f}(\hat{x})^{\ast}$.
\end{itemize}
\end{myprop}

\subsection{Useful notation}

This subsection introduces useful notation that will be used throughout this paper. The usefulness of this notation will be clear later. We added this subsection so that the reader can refer to it anytime.

Let $W:G_1\times G_2\longrightarrow \mathcal{Z}$ be a two-user MAC and let $(X,Y)\stackrel{W}{\longrightarrow}Z$. Define the following:
\begin{itemize}
\item $\YZ(W):=\{(y,z)\in G_2\times\mathcal{Z}:\; P_{Y,Z}(y,z)>0\}$. This is just the support of $P_{Y,Z}$.
\item For every $(y,z)\in \YZ(W)$, define $p_{y,z,W}:G_1\rightarrow[0,1]$ as $p_{y,z,W}(x)=P_{X|Y,Z}(x|y,z)$ for every $x\in G_1$.
\end{itemize}

For every $z\in\mathcal{Z}$, define:
\begin{itemize}
\item $\Y^z(W):=\{y\in G_2:\;P_{Y,Z}(y,z)>0\}$.
\item ${\textstyle\DY^z(W)}:=\big\{y_1-y_2:\;y_1,y_2\in \textstyle\Y^z(W)\big\}$.
\item $\textstyle\X^z(W):=\big\{\hat{x}\in G_1:\exists y\in\textstyle\Y^z(W),\hat{p}_{y,z,W}(\hat{x})\neq 0\big\}$.
\item $\XDY^z(W):=\X^z(W)\times \DY^z(W)=\big\{(\hat{x},y):\;\hat{x}\in\X^z(W),\;y\in\DY^z(W)\big\}$.
\end{itemize}

Now define:
\begin{itemize}
\item $\textstyle\XZ(W):=\big\{(\hat{x},z):\;z\in\mathcal{Z},\;\hat{x}\in\textstyle\X^z(W)\big\}.$
\item $\XDY(W):=\displaystyle\bigcup_{z\in\mathcal{Z}}\textstyle\XDY^z(W)$.
\end{itemize}

\subsection{Pseudo-quadratic functions}

\begin{mydef}
\label{defQuadSub}
Let $D\subset G_1\times G_2$. Define the following sets:
\begin{itemize}
\item $H_1(D):=\{x\in G_1:\;\exists y\in G_2,\;(x,y)\in D\}$.
\item For every $x\in H_1(D)$, let $H_2^{x}(D):=\{y\in G_2:\;(x,y)\in D\}$.
\item $H_2(D):=\{y\in G_2:\;\exists x\in G_1,\;(x,y)\in D\}$.
\item For every $y\in H_2(D)$, let $H_1^{y}(D):=\{x\in G_1:\;(x,y)\in D\}$.
\end{itemize}
We say that $D$ is a \emph{pseudo-quadratic domain} if:
\begin{itemize}
\item $H_1^{y}(D)$ is a subgroup of $G_1$ for every $y\in H_2(D)$.
\item $H_2^{x}(D)$ is a subgroup of $G_2$ for every $x\in H_1(D)$.
\end{itemize}
\end{mydef}

\begin{mydef}
\label{defQuad}
Let $D\subset G_1\times G_2$ and let $F:D\rightarrow\mathbb{T}$ be a mapping from $D$ to $\mathbb{T}=\{\omega\in\mathbb{C}:\;|\omega|=1\}$. We say that $F$ is a \emph{pseudo-quadratic function} if:
\begin{itemize}
\item $D$ is a pseudo-quadratic domain.
\item For every $y\in H_2(D)$, the mapping $x\rightarrow F(x,y)$ is a group homomorphism from $\big(H_1^y(D),+\big)$ to $(\mathbb{T},\cdot)$.
\item For every $x\in H_1(D)$, the mapping $y\rightarrow F(x,y)$ is a group homomorphism from $\big(H_2^x(D),+\big)$ to $(\mathbb{T},\cdot)$.
\end{itemize}
\end{mydef}

\begin{mydef}
\label{defComp}
We say that $W:G_1\times G_2\longrightarrow\mathcal{Z}$ is \emph{polarization compatible with respect to the first user} if there exists a pseudo-quadratic function $F:D\rightarrow \mathbb{T}$ such that:
\begin{itemize}
\item $\XDY(W)\subset D\subset G_1\times G_2$.
\item For every $(\hat{x},z)\in \XZ(W)$ and every $y_1,y_2\in\Y^z(W)$, we have $\hat{p}_{y_1,z,W}(\hat{x})=F(\hat{x},y_1-y_2)\cdot\hat{p}_{y_2,z,W}(\hat{x})$.
\end{itemize}
\end{mydef}

\subsection{Main result}

The following theorem is the main result of this paper:
\begin{mythe}
\label{MainThe}
If $W$ is a two-user MAC, then polarization $\ast$-preserves $I_1$ for $W$ if and only if $W$ is polarization compatible with respect to the first user.
\end{mythe}

Theorem \ref{MainThe} has the following implications:
\begin{itemize}
\item (Proposition \ref{propSimple}) If $G_1=G_2=\mathbb{F}_q$ for a prime $q$ and $(X,Y)\stackrel{W}{\longrightarrow}Z$, then polarization $\ast$-preserves $I_1$ for $W$ if and only if there exists $a\in\mathbb{F}_q$ such that $I(X+aY;Y|Z)=0$.
\item (Corollary \ref{corBAC}) Polarization $\ast$-preserves the symmetric capacity region for the binary adder channel.
\item (Proposition \ref{propDif}) If $|G_1|$ and $|G_2|$ are co-prime and $(X,Y)\stackrel{W}{\longrightarrow}Z$, then polarization $\ast$-preserves $I_1$ for $W$ if and only if $I(X;Y|Z)=0$ (i.e., if and only if the dominant face of $\mathcal{J}(W)$ is a single point).
\end{itemize}

\vspace*{3mm}

The reader may find the polarization compatibility condition (Definition \ref{defComp}) too abstract at this stage and it may not be clear why the $\ast$-preservation of $I_1$ has anything to do with pseudo-quadratic functions. In order to clarify the meaning of polarization compatibility and make it more intuitive, we provide in Section \ref{secSuffCond} a sufficient condition for the $\ast$-preservation of $I_1$ that is easy to understand. After expressing this condition in terms of $\{\hat{p}_{y,z,W}:\;(y,z)\in\YZ(W)\}$, the link between the $\ast$-preservation of $I_1$ and pseudo-quadratic functions should become clear.

\section{A sufficient condition for the $\ast$-preservation of $I_1$}
\label{secSuffCond}
In this section, we only consider two-user MACs $W:G_1\times G_2\longrightarrow \mathcal{Z}$, where $G_1$ and $G_2$ are finite Abelian groups. We derive a sufficient condition which ensures that polarization $\ast$-preserves $I_1$.

\begin{mydef}
Let $W:G_1\times G_2\longrightarrow \mathcal{Z}$ be a two-user MAC. We say that $I_1$ is \emph{preserved} for $W$ if and only if $I_1(W^-)+I_1(W^+)=2I_1(W)$.
\end{mydef}

\begin{mylem}
\label{lemRecPres}
Polarization $\ast$-preserves $I_1$ for $W$ if and only if $I_1$ is preserved for $W^s$ for every $s\in\{-,+\}^{\ast}$.
\end{mylem}
\begin{proof}
Polarization $\ast$-preserves $I_1$ for $W$ if and only if
\begin{align*}
\forall n\geq 0,\; I_1(W)=\frac{1}{2^n}\sum_{s\in\{-,+\}^n}&I_1(W^s)\;\;\Leftrightarrow\;\;\forall n\geq 0,\; \frac{1}{2^n}\sum_{s\in\{-,+\}^n}I_1(W^s)=\frac{1}{2^{n+1}}\sum_{s'\in\{-,+\}^{n+1}}I_1(W^{s'})\\
&\Leftrightarrow\;\;\forall n\geq 0,\; \sum_{s\in\{-,+\}^n}2I_1(W^s)=\sum_{s\in\{-,+\}^n}(I_1(W^{(s,-)})+I_1(W^{(s,+)}))\\
&\Leftrightarrow\;\;\forall n\geq 0,\; \sum_{s\in\{-,+\}^n}\big(2I_1(W^s)-I_1(W^{(s,-)})-I_1(W^{(s,+)})\big)=0.
\end{align*}

But since $2I_1(W^s)-I_1(W^{(s,-)})-I_1(W^{(s,+)})\geq 0$ (apply \eqref{mainineq} to $W^s$), we conclude that polarization $\ast$-preserves $I_1$ for $W$ if and only if
$\forall n\geq 0,\forall s\in\{-,+\}^n,\; I_1(W^{(s,-)})+I_1(W^{(s,+)})=2I_1(W^s)$. In other words, polarization $\ast$-preserves $I_1$ for $W$ if and only if $I_1$ is preserved for $W^s$ for every $s\in\{-,+\}^{\ast}$.
\end{proof}

\vspace*{3mm}

Suppose we want to prove that a given condition on $W$ is sufficient for the $\ast$-preservation of $I_1$. Lemma \ref{lemRecPres} suggests a method to do that:
\begin{enumerate}
\item Show that if $W$ satisfies the condition, then $I_1$ is preserved for $W$.
\item Show that if $W$ satisfies the condition, then $W^-$ and $W^+$ satisfies the condition as well.
\end{enumerate}

By doing that, we would have shown that if $W$ satisfies the condition, then $W^s$ satisfies the same condition for all $s\in\{-,+\}^{\ast}$, which in turn implies that $I_1$ is preserved for $W^s$ for all $s\in\{-,+\}^{\ast}$, hence polarization $\ast$-preserves $I_1$ for $W$ due to Lemma \ref{lemRecPres}.

\begin{mydef}
\label{DefInd}
Let $W:G_1\times G_2\longrightarrow \mathcal{Z}$ be a two-user MAC and let $(X,Y)\stackrel{W}{\longrightarrow}Z$. We say that $W$ is homomorphic-independent with respect to the first user if and only if there exists a subgroup $H_2$ of $G_2$, a group homomorphism $f:H_2\rightarrow G_1$ and a mapping $g:\mathcal{Z}\rightarrow G_2$ such that:
\begin{itemize}
\item $\mathbb{P}[Y-g(Z)\in H_2]=1$.
\item $I\big(X+f(Y-g(Z));Y\big|Z\big)=0$.
\end{itemize}
\end{mydef}

The condition $\mathbb{P}[Y-g(Z)\in H_2]=1$ means that $Y$ and $g(Z)$ belong to the same coset of $H_2$. In other words, given $Z=z$, $Y$ belongs to a single coset of $H_2$, and this coset is determined by $g(z)$. On the other hand, the condition $I\big(X+f(Y-g(Z));Y\big|Z\big)=0$ is equivalent to say that given $Z$, a shifted version of $X$ is conditionally independent of $Y$, and the amount by which $X$ should be shifted is $f(Y-g(Z))$. One might be tempted to simplify the expression $I\big(X+f(Y-g(Z));Y\big|Z\big)$ as follows:
$$I\big(X+f(Y-g(Z));Y\big|Z\big)=I\big(X+f(Y)-f(g(Z));Y\big|Z\big)=I\big(X+f(Y);Y\big|Z\big).$$
This would be correct if $f$ were defined on the whole group $G_2$. However, $f$ is only defined on a subgroup $H_2$ of $G_2$. This is why $f$ can be applied on $Y-g(Z)$ which belongs to $H_2$, but cannot be applied to $Y$ and $g(Z)$ individually because they can lie outside $H_2$.

In the rest of this section, we show that if $W$ is homomorphic-independent with respect to the first user, then polarization $\ast$-preserves $I_1$ for $W$. For the sake of brevity and simplicity, we will simply write homomorphic-independent to denote ``homomorphic-independent with respect to the first user".

\begin{mylem}
\label{lemSuf1First}
If $W:G_1\times G_2\longrightarrow \mathcal{Z}$ is homomorphic-independent, then $I_1$ is preserved for $W$.
\end{mylem}
\begin{proof}
Let $U_1,U_2,V_1,V_2,X_1,X_2,Y_1,Y_2,Z_1,Z_2$ be as in Remark \ref{MainRem}. We can see from \eqref{mainineq} that $I_1$ is preserved for $W$ if and only if $I(U_1;Z_1Z_2V_1)=I(U_1;Z_1Z_2V_1V_2)$. Therefore, it is sufficient to show that $I(U_1;V_2|Z_1Z_2V_1)=0$.

Let $H_2$, $f$ and $g$ be as in Definition \ref{DefInd}. We have:
$$V_1-g(Z_1)+g(Z_2)=Y_1-Y_2-g(Z_1)+g(Z_2)=\big(Y_1-g(Z_1)\big)-\big(Y_2-g(Z_2)\big)\stackrel{(a)}{\in} H_2,$$
where (a) is true because $\mathbb{P}[Y_1-g(Z_1)\in H_2]=\mathbb{P}[Y_2-g(Z_2)\in H_2]=1$.

Let $\tilde{X}_1=X_1+f(Y_1-g(Z_1))$ and $\tilde{X}_2=X_2+f(Y_2-g(Z_2))$. We have:
\begin{equation}
\label{eqrefeq}
\begin{aligned}
U_1+f\big(V_1-g(Z_1)+g(Z_2)\big)&=X_1-X_2+f(Y_1-Y_2-g(Z_1)+g(Z_2))\\
&=X_1+f(Y_1-g(Z_1))-X_2-f(Y_2-g(Z_2))=\tilde{X}_1-\tilde{X}_2.
\end{aligned}
\end{equation}
Therefore,
\begin{align*}
I(U_1;V_2|Z_1Z_2V_1)&=I\big(U_1-f\big(V_1-g(Z_1)+g(Z_2)\big);V_2\big|Z_1Z_2V_1\big)=
I(\tilde{X}_1-\tilde{X}_2;V_2|Z_1Z_2V_1)\\
&\leq I(\tilde{X}_1\tilde{X}_2;V_2|Z_1Z_2V_1)\leq I(\tilde{X}_1\tilde{X}_2;V_1V_2|Z_1Z_2)=I(\tilde{X}_1\tilde{X}_2;Y_1Y_2|Z_1Z_2)\\
&=I(\tilde{X}_1;Y_1|Z_1)+I(\tilde{X}_2;Y_2|Z_2)\stackrel{(b)}{=}0,
\end{align*}
where (b) follows from the fact that $W$ is homomorphic-independent. We conclude that $I(U_1;V_2|Z_1Z_2V_1)=0$ and so $I_1$ is preserved for $W$.
\end{proof}

\begin{mylem}
\label{lemSuf2First}
If $W:G_1\times G_2\longrightarrow \mathcal{Z}$ is homomorphic-independent, then $W^-$ and $W^+$ are homomorphic-independent as well.
\end{mylem}
\begin{proof}
Let $U_1,U_2,V_1,V_2,X_1,X_2,Y_1,Y_2,Z_1,Z_2$ be as in Remark \ref{MainRem}. Let $H_2$, $f$ and $g$ be as in Definition \ref{DefInd}. Define the mappings $g^-:\mathcal{Z}^2\rightarrow G_2$ and $g^+:\mathcal{Z}^2\times G_1\times G_2\rightarrow G_2$ as follows:
$$g^-(z_1,z_2)=g(z_1)-g(z_2)\;\;\text{and}\;\;g^+(z_1,z_2,u_1,v_1)=g(z_2).$$

Since $W$ is homomorphic-independent, we have $\mathbb{P}[Y_1-g(Z_1)\in H_2]=1$ and $\mathbb{P}[Y_2-g(Z_2)\in H_2]=1$. Therefore, $\mathbb{P}[Y_1-Y_2-g(Z_1)+g(Z_2)\in H_2]=1$ which implies that $\mathbb{P}\big[V_1-g^-(Z_1,Z_2)\in H_2\big]=1$. Similarly, $\mathbb{P}\big[V_2-g^+(Z_1,Z_2,U_1,V_1)\in H_2\big]=\mathbb{P}[Y_2-g(Z_2)\in H_2]=1$.

Define $\tilde{X}_1=X_1+f(Y_1-g(Z_1))$ and $\tilde{X}_2=X_2+f(Y_2-g(Z_2))$ as in the proof of Lemma \ref{lemSuf1First}. From \eqref{eqrefeq} we have $U_1+f(V_1-g^-(Z_1,Z_2)) =\tilde{X}_1-\tilde{X}_2$. Therefore,
\begin{align*}
I\big(U_1+f(V_1-g^-(Z_1,Z_2));V_1\big|Z_1Z_2\big)&=I(\tilde{X}_1-\tilde{X}_2;V_1|Z_1Z_2)\leq I(\tilde{X}_1\tilde{X}_2;V_1V_2|Z_1Z_2)\\
&=I(\tilde{X}_1\tilde{X}_2;Y_1Y_2|Z_1Z_2)=I(\tilde{X}_1;Y_1|Z_1)+I(\tilde{X}_2;Y_2|Z_2)\stackrel{(a)}{=}0,
\end{align*}
where (a) follows from the fact that $W$ is homomorphic-independent.

On the other hand, we have
\begin{align*}
I\big(U_2+f&(V_2-g^+(Z_1,Z_2,U_1,V_1));V_2\big|Z_1Z_2U_1V_1\big)\\
&=I\big(X_2+f(Y_2-g(Z_2));V_2\big|Z_1Z_2U_1V_1\big)=I(\tilde{X}_2;V_2|Z_1Z_2,U_1+f(V_1-g^-(Z_1,Z_2)),V_1)\\
&=I(\tilde{X}_2;V_2|Z_1Z_2,\tilde{X}_1-\tilde{X}_2,V_1)\leq I(\tilde{X}_2,\tilde{X}_1-\tilde{X}_2;V_2V_1|Z_1Z_2)=I(\tilde{X}_1\tilde{X}_2;Y_1Y_2|Z_1Z_2)\\
&=I(\tilde{X}_1;Y_1|Z_1)+I(\tilde{X}_2;Y_2|Z_2)=0.
\end{align*}

We conclude that $W^-$ and $W^+$ are homomorphic-independent.
\end{proof}

\begin{myprop}
\label{propSufFirst}
If $W:G_1\times G_2\longrightarrow \mathcal{Z}$ is homomorphic-independent, then polarization $\ast$-preserves $I_1$ for $W$.
\end{myprop}
\begin{proof}
We first show by induction on $n\geq0$ that for every $s\in\{-,+\}^n$, $W^s$ is homomorphic-independent. If $n=0$, there is nothing to prove. Now let $n>0$ and suppose that the claim is true for $n-1$.

Let $s\in\{-,+\}^n$, then there exists $s'\in\{-,+\}^{n-1}$ such that $s=(s',-)$ or $s=(s',+)$. We know from the induction hypothesis that $W^{s'}$ is homomorphic-independent, and by applying Lemma \ref{lemSuf2First} to $W^{s'}$ we deduce that both $W^{(s',-)}$ and $W^{(s',+)}$ are homomorphic-independent. Therefore, $W^s$ is homomorphic-independent. 

We conclude that for every $n\geq0$ and every $s\in\{-,+\}^n$, $W^s$ is homomorphic-independent. Lemma \ref{lemSuf1First} implies that $I_1$ is preserved for $W^s$ for every $s\in \{-,+\}^s$, and Lemma \ref{lemRecPres} shows that polarization $\ast$-preserves $I_1$ for $W$.
\end{proof}

\vspace*{3mm}

One might try to simplify the sufficient condition that we have just shown in the following way. If $H_2$, $f$ and $g$ are as in Definition \ref{DefInd}, let $\tilde{f}: G_2\rightarrow G_1$ be an extension of $f$ which is a homomorphism from $(G_2,+)$ to $(G_1,+)$. We have:
\begin{align*}
I\big(X+\hat{f}(Y);Y\big|Z\big)&=I\big(X+\hat{f}(Y)-\hat{f}(g(Z));Y\big|Z\big)\\
&=I\big(X+\hat{f}(Y-g(Z));Y\big|Z\big)=I\big(X+f(Y-g(Z));Y\big|Z\big)=0.
\end{align*}
This would suggest that homomorphic-independence is equivalent to the existence of a homomorphism $\hat{f}:G_2\to G_1$ satisfying $I(X+ \hat{f}(Y);Y|Z)=0$, which is of course simpler than the way homomorphic-independence was defined in Definition \ref{DefInd}. This argument breaks down when we realize that not every homomorphism $f:H_2\to G_1$ can be extended to a homomorphism from $(G_2,+)$ to $(G_1,+)$. For example, if $G_1=\mathbb{F}_2$, $G_2=\mathbb{Z}_4$, $H_2=\{0,2\}\subset G_2$ and $f: H_2\to G_1$ is defined as $f(0)=0$ and $f(2)=1$, then $f$ is clearly a homomorphism from $H_2$ to $G_1$. However, $f$ is not extendable to a homomorphism $\hat{f}:G_2\to G_1$ defined on the whole group $G_2$. If $f$ were extendable, we would have $1=\hat{f}(2)=\hat{f}(1)+\hat{f}(1)=2\hat{f}(1)=0$, which is a contradiction.

The existence of a homomorphism $f:G_2\to G_1$ satisfying $I(X+f(Y);Y|Z)$ is of course a sufficient condition for the $\ast$-preservation of $I_1$ because it is a particular case of homomorphic-independence. However, homomorphic-independence is a strictly more general condition as we have shown in the previous paragraph.

Note that there is a large freedom on the choice of the mapping $g:\mathcal{Z}\to G_2$ in Definition \ref{DefInd}. The main role of the mapping $g$ is to find the coset of $H_2$ to which $Y$ belongs, and any other mapping $g'$ playing this role will satisfy the conditions of Definition \ref{DefInd}: Let $H_2$, $f$ and $g$ be as in Definition \ref{DefInd} and assume that $g':\mathcal{Z}\to G_2$ satisfies $g'(z)-g(z)\in H_2$ for all $z\in\mathcal{Z}$. We have $Y-g'(Z)=Y-g(Z)+g(Z)-g'(Z)\in H_2$ with probability 1. On the other hand,
\begin{align*}
I\big(X+f(Y-g'(Z));Y\big|Z\big)&=I\big(X+f\big(Y-g(Z)+g(Z)-g'(Z)\big);Y\big|Z\big)\\
&=I\big(X+f(Y-g(Z))+f\big(g(Z)-g'(Z)\big);Y\big|Z\big)\\
&=I\big(X+f(Y-g(Z));Y\big|Z\big)=0.
\end{align*}
Therefore, $H_2$, $f$ and $g'$ also satisfy the conditions of Definition \ref{DefInd}.

Let us now see how homomorphic-independence can be expressed in terms of $\{\hat{p}_{y,z,W}:\;(y,z)\in\YZ(W)\}$. For the sake of simplicity and brevity, we will simply write $p_{y,z}$ to denote $p_{y,z,W}$.

The condition $I\big(X+f(Y-g(Z));Y\big|Z\big)=0$ is equivalent to the conditional independence of $X+f(Y-g(Z))$ and $Y$ given $Z$. This is equivalent to say that for every $x\in G_1$, every $y_1,y_2\in G$ and every $z\in \mathcal{Z}$ satisfying $P_{Y,Z}(y_1,z)>0$ and $P_{Y,Z}(y_2,z)>0$, we have 
$$P_{X+f(Y-g(Z))|Y,Z}\big(x\big|y_1,z\big)
=P_{X+f(Y-g(Z))|Y,Z}\big(x\big|y_2,z\big).$$
On the other hand, we have
$$P_{X+f(Y-g(Z))|Y,Z}\big(x\big|y_1,z\big)=P_{X|Y,Z}\big(x-f(y_1-g(z))\big|y_1,z\big)=p_{y_1,z}(x-f(y_1-g(z))),$$
and
$$P_{X+f(Y-g(Z))|Y,Z}\big(x\big|y_2,z\big)=P_{X|Y,Z}\big(x-f(y_2-g(z))\big|y_2,z\big)=p_{y_2,z}(x-f(y_2-g(z))).$$
Therefore, the condition $I\big(X+f(Y-g(Z));Y\big|Z\big)=0$ is equivalent to say that for every $z\in \mathcal{Z}$ and every $y_1,y_2\in \Y^z(W)$, we have 
\begin{align*}
&p_{y_1,z}(x-f(y_1-g(z)))=p_{y_2,z}(x-f(y_2-g(z)))\;\; \forall x\in G_1\\
\Leftrightarrow\;\;&p_{y_1,z}(x)=p_{y_2,z}\big(x+f(y_1-g(z))-f(y_2-g(z))\big)\;\; \forall x\in G_1\\
\Leftrightarrow\;\;&p_{y_1,z}(x)=p_{y_2,z}\big(x+f(y_1-g(z)-y_2+g(z))\big)\;\; \forall x\in G_1\\
\Leftrightarrow\;\;&p_{y_1,z}(x)=p_{y_2,z}\big(x+f(y_1-y_2)\big)\;\; \forall x\in G_1.
\end{align*}

This shows that if we want to get rid of the mapping $g$ in the second condition of Definition \ref{DefInd}, we have to express the homomorphic-independence condition in terms of the conditional probability distributions $\{p_{y,z}:\;(y,z)\in\YZ(W)\}$. On the other hand, due to the freedom on the choice of the mapping $g$ that we have shown above, we can see that $g$ in the condition $\mathbb{P}[Y-g(Z)\in H_2]$ just serves the purpose of saying that given $Z=z$, $Y$ belongs to a single coset of $H_2$. In other words, $y_1-y_2\in H_2$ for all $z\in\mathcal{Z}$ and every $y_1,y_2\in \Y^z(W)$, which is equivalent to say that $\DY^z(W)\subset H_2$ for all $z\in \mathcal{Z}$. Therefore, the first condition of Definition \ref{DefInd} can be replaced by $\XDY(W)\subset G_1\times H_2$. This shows the following lemma:

\begin{mylem}
\label{lemCharacInd1}
$W$ is homomorphic-independent if and only if there exists a subgroup $H_2$ of $G_2$ and a homomorphism $f: H_2\to G_1$ satisfying:
\begin{itemize}
\item $\XDY(W)\subset G_1\times H_2$.
\item For every $z\in\mathcal{Z}$ and every $y_1,y_2\in\Y^z(W)$, we have:
\begin{align*}
&p_{y_1,z}(x)=p_{y_2,z}(x+f(y_1-y_2))\;\; \forall x\in G_1\\
\Leftrightarrow\;\;&\hat{p}_{y_1,z}(\hat{x})=\hat{p}_{y_2,z}\big(\hat{x}\big)e^{j2\pi\langle\hat{x},f(y_1-y_2)\rangle}\;\; \forall \hat{x}\in G_1.
\end{align*}
\end{itemize}
\end{mylem}

Lemma \ref{lemCharacInd1} suggests that if we are given a two-user MAC $W$ and we want to check whether it is homomorphic-independent, then one way to do that is to compute $Q(\hat{x},y_1,y_2,z)=\frac{\hat{p}_{y_1,z}(\hat{x})}{\hat{p}_{y_2,z}(\hat{x})}$ for every $z\in\mathcal{Z}$, every $y_1,y_2\in\Y^z(W)$ and every $\hat{x}$ satisfying $\hat{p}_{y_2,z}(\hat{x})\neq 0$, and then make sure that $Q(\hat{x},y_1,y_2,z)$ can be expressed as $e^{j2\pi\langle\hat{x},f(y_1-y_2)\rangle}$ for some homomorphism $f:H_2\to G_1$, where $H_2$ is a subgroup of $G_2$ that satisfies $\XDY(W)\subset G_1\times H_2$.

We can now make the following remarks:
\begin{itemize}
\item $e^{j2\pi\langle\hat{x},f(y_1-y_2)\rangle}\in\mathbb{T}:=\{\omega\in\mathbb{C}:\;|\omega|=1\}.$
\item $e^{j2\pi\langle\hat{x},f(y_1-y_2)\rangle}$ depends only on $\hat{x}$ and $y_1-y_2$.
\item For every $y\in H_2$, the mapping $\hat{x}\rightarrow e^{j2\pi\langle\hat{x},f(y)\rangle}$ is a group homomorphism from $(G_1,+)$ to $(\mathbb{T},\cdot)$.
\item For every $\hat{x}\in G_1$, the mapping $y\rightarrow e^{j2\pi\langle\hat{x},f(y)\rangle}$ is a group homomorphism from $(H_2,+)$ to $(\mathbb{T},\cdot)$.
\end{itemize}

Therefore, the mapping $(\hat{x},y)\to  e^{j2\pi\langle\hat{x},f(y)\rangle}$ is a pseudo-quadratic function from $G_1\times H_2$ to $\mathbb{T}$.

We can now show the following characterization of homomorphic-independent MACs:

\begin{myprop}
\label{propCharacInd}
Let $W: G_1\times G_2\longrightarrow \mathcal{Z}$ be a two-user MAC. $W$ is homomorphic-independent if and only if there exists a subgroup $H_2$ of $G_2$ and a pseudo-quadratic function $F:G_1\times H_2\to\mathbb{T}$ satisfying:
\begin{itemize}
\item $\XDY(W)\subset G_1\times H_2$.
\item For every $(\hat{x},z)\in \XZ(W)$ and every $y_1,y_2\in\Y^z(W)$, we have $\hat{p}_{y_1,z}(\hat{x})=F(\hat{x},y_1-y_2)\hat{p}_{y_2,z}(\hat{x})$.
\end{itemize}
\end{myprop}
\begin{proof}
The above discussion shows that the existence of such $H_2$ and $F$ is a necessary condition for the homomorphic-independence of $W$. For the proof that it is also sufficient, see Appendix \ref{appCharacInd}.
\end{proof}

\vspace*{3mm}

Note that the only difference between polarization compatibility (Definition \ref{defComp}) and the characterization of homomorphic-independence of Proposition \ref{propCharacInd} is that the domain $D$ of the pseudo-quadratic function $F$ in Definition \ref{defComp} can be an arbitrary pseudo-quadratic domain, whereas the domain of $F$ in Proposition \ref{propCharacInd} needs to be of the form $G_1\times H_2$ for some subgroup $H_2$ of $G_2$. In the next section, we show that polarization compatibility is a necessary and sufficient condition for the $\ast$-preservation of $I_1$.

\section{Two-user MACs with $\ast$-preserved $I_1$}
\label{secNecSufCond2UserMac}

Throughout this section, we fix a two-user MAC $W:G_1\times G_2\longrightarrow \mathcal{Z}$, where $G_1$ and $G_2$ are finite Abelian groups. This section is dedicated to prove Theorem \ref{MainThe}.

\subsection{Polarization compatibility is necessary}

For the sake of simplicity, we write $p_{y,z}(x)$ to denote $p_{y,z,W}(x)$.

According to \eqref{mainineq}, $I_1$ is preserved for $W$ if and only if $I(U_1;V_2|Z_1Z_2V_1)=0$, which means that for every $z_1,z_2\in\mathcal{Z}$ and every $v_1,v_2\in G_2$, if $P_{V_2,Z_1,Z_2,V_1}(v_2,z_1,z_2,v_1)>0$ then $P_{U_1|V_2,Z_1,Z_2,V_1}(u_1|v_2,z_1,z_2,v_1)$ does not depend on $v_2$.

In order to study this condition, we should keep track of the values of $z_1,z_2\in\mathcal{Z}$ and $v_1,v_2\in G_2$ for which $P_{V_2,Z_1,Z_2,V_1}(v_2,z_1,z_2,v_1)>0$. But $P_{V_2,Z_1,Z_2,V_1}(v_2,z_1,z_2,v_1)=P_{Y_1,Z_1}(v_1+v_2,z_1)P_{Y_2,Z_2}(v_2,z_2)$, so it is sufficient to keep track of the pairs $(y,z)\in G_2\times\mathcal{Z}$ satisfying $P_{Y,Z}(y,z)>0$. This is where $\YZ(W)$ and $\{\Y^z(W):\;z\in\mathcal{Z}\}$ become useful.

%

The following lemma gives a characterization of two-user MACs with preserved $I_1$ in terms of the Fourier transform of the distributions $p_{y,z}$.

\begin{mylem}
\label{lemLemlemLem}
$I_1$ is preserved for $W$ if and only if for every $y_1,y_2,y_1',y_2'\in G_2$ and every $z_1,z_2\in\mathcal{Z}$ satisfying
\begin{itemize}
\item $y_1-y_2=y_1'-y_2'$,
\item $y_1,y_1'\in\Y^{z_1}(W)$ and $y_2,y_2'\in\Y^{z_2}(W)$,
\end{itemize}
we have
$$\hat{p}_{y_1,z_1}(\hat{x})\cdot\hat{p}_{y_2,z_2}(\hat{x})^{\ast} =\hat{p}_{y_1',z_1}(\hat{x}) \cdot\hat{p}_{y_2',z_2}(\hat{x})^{\ast},\;\;\forall\hat{x}\in G_1.$$
\end{mylem}
\begin{proof}
Let $U_1,U_2,V_1,V_2,X_1,X_2,Y_1,Y_2,Z_1,Z_2$ be as in Remark \ref{MainRem}. We know that $I_1$ is preserved for $W$ if and only if $I(U_1;V_2|Z_1Z_2V_1)=0$, which is equivalent to say that $U_1$ is conditionally independent of $V_2$ given $(Z_1,Z_2,V_1)$.

In other words, for any fixed $(z_1,z_2,v_1)\in\mathcal{Z}\times\mathcal{Z}\times G_2$ satisfying $P_{Z_1,Z_2,V_1}(z_1,z_2,v_1)>0$, if $v_2,v_2'\in G_2$ satisfy $P_{V_2|Z_1,Z_2,V_1}(v_2|z_1,z_2,v_1)>0$ and $P_{V_2|Z_1,Z_2,V_1}(v_2'|z_1,z_2,v_1)>0$, then we have
\begin{align*}
\forall u_1\in G_1,\;P_{U_1|V_2,Z_1,Z_2,V_1}(u_1|v_2,z_1,z_2,v_1)=P_{U_1|V_2,Z_1,Z_2,V_1}(u_1|v_2',z_1,z_2,v_1).
\end{align*}
This condition is equivalent to saying that, for every $z_1,z_2\in\mathcal{Z}$ and every $v_1,v_2,v_2'\in G_2$ satisfying $P_{Z_1,Z_2,Y_1,Y_2}(z_1,z_2,v_1+v_2,v_2)>0$ and $P_{Z_1,Z_2,Y_1,Y_2}(z_1,z_2,v_1+v_2',v_2')>0$, we have
\begin{align*}
\forall u_1\in G_1,\;&P_{X_1-X_2|Z_1,Z_2,Y_1,Y_2}(u_1|z_1,z_2,v_1+v_2,v_2)=P_{X_1-X_2|Z_1,Z_2,Y_1,Y_2}(u_1|z_1,z_2,v_1+v_2',v_2').
\end{align*}
By denoting $v_1+v_2,v_2,v_1+v_2'$ and $v_2'$ as $y_1,y_2,y_1'$ and $y_2'$ respectively (so that $y_1-y_2=y_1'-y_2'=v_1$), we can deduce that $I_1$ is preserved for $W$ if and only if for every $y_1,y_2,y_1',y_2'\in G_2$ and every $z_1,z_2\in\mathcal{Z}$ satisfying $y_1-y_2=y_1'-y_2'$, $P_{Z_1,Z_2,Y_1,Y_2}(z_1,z_2,y_1,y_2)>0$ and $P_{Z_1,Z_2,Y_1,Y_2}(z_1,z_2,y_1',y_2')>0$ (i.e., $y_1,y_1'\in\Y^{z_1}(W)$ and $y_2,y_2'\in\Y^{z_2}(W)$), we have
\begin{align*}
\forall u_1\in G_1,\;&P_{X_1-X_2|Z_1,Z_2,Y_1,Y_2}(u_1|z_1,z_2,y_1,y_2)=P_{X_1-X_2|Z_1,Z_2,Y_1,Y_2}(u_1|z_1,z_2,y_1',y_2').
\end{align*}

On the other hand, we have:
\begin{align*}
P_{X_1-X_2|Z_1,Z_2,Y_1,Y_2}(u_1|z_1,z_2,y_1,y_2)&=\sum_{u_2\in G_1}P_{X_1|Z_1,Y_1}(u_1+u_2|z_1,y_1)P_{X_2|Z_2,Y_2}(u_2|z_2,y_2)\\
&=\sum_{u_2\in G_1} p_{y_1,z_1}(u_1+u_2)p_{y_2,z_2}(u_2)=(p_{y_1,z_1}\ast \tilde{p}_{y_2,z_2})(u_1),
\end{align*}
where we define $\tilde{p}_{y_2,z_2}(x):=p_{y_2,z_2}(-x)$. Similarly $P_{X_1-X_2|Z_1,Z_2,Y_1,Y_2}(u_1|z_1,z_2,y_1',y_2')=(p_{y_1',z_1}\ast \tilde{p}_{y_2',z_2})(u_1)$. Therefore, for every $u_1\in G_1$, we have $$(p_{y_1,z_1}\ast \tilde{p}_{y_2,z_2})(u_1)=(p_{y_1',z_1}\ast \tilde{p}_{y_2',z_2})(u_1),$$
which is equivalent to $\hat{p}_{y_1,z_1}(\hat{u}_1)\cdot\hat{p}_{y_2,z_2}(\hat{u}_1)^{\ast}=\hat{p}_{y_1',z_1}(\hat{u}_1) \cdot\hat{p}_{y_2',z_2}(\hat{u}_1)^{\ast}$ for every $\hat{u}_1\in G_1$.
\end{proof}

\begin{mydef}
Let $W:G_1\times G_2\longrightarrow \mathcal{Z}$ be a two-user MAC. We say that $I_1$ is \emph{$\ast^-$ preserved} for $W$ if and only if $I_1$ is preserved for $W^{[n]^-}$ for every $n\geq 0$, where $[n]^-\in\{-,+\}^n$ is the sequence containing $n$ minus signs (e.g., $[0]^-=\o$, $[2]^-=(-,-)$).
\end{mydef}


The following three lemmas study the MACs $W$ for which $I_1$ is $\ast^-$ preserved.

\begin{mylem}
\label{lem2}
If $I_1$ is $\ast^-$ preserved for $W$, then for every $n>0$, every $y_1,\ldots,y_{2^n},y_1',\ldots, y_{2^n}'\in G_2$ and every $z_1,\ldots, z_{2^n}\in\mathcal{Z}$ satisfying
\begin{itemize}
\item $\displaystyle \sum_{i=1}^{2^n} y_i=\sum_{i=1}^{2^n} y_i'$,
\item $y_1\in\Y^{z_1}(W),\ldots, y_{2^n}\in\Y^{z_{2^n}}(W)$, and
\item $y_1'\in\Y^{z_1}(W),\ldots, y_{2^n}'\in\Y^{z_{2^n}}(W)$,
\end{itemize}
we have
$$\prod_{i=1}^{2^n} \hat{p}_{y_i,z_i}(\hat{x})=\prod_{i=1}^{2^n} \hat{p}_{y_i',z_i}(\hat{x}),\;\;\forall\hat{x}\in G_1.$$
\end{mylem}
\begin{proof}
See Appendix \ref{applem2}.
\end{proof}

\begin{mylem}
\label{lem3}
If $I_1$ is $\ast^-$ preserved for $W$ then for every $(\hat{x},z)\in \XZ(W)$, we have:
\begin{itemize}
\item $\hat{p}_{y,z}(\hat{x})\neq 0$ for all $y\in \Y^{z}(W)$.
\item $\displaystyle\frac{\hat{p}_{y,z}(\hat{x})}{\hat{p}_{y',z}(\hat{x})}\in\mathbb{T}$ for every $y,y'\in\Y^z(W)$, where $\mathbb{T}:=\{\omega\in\mathbb{C}:\;|\omega|=1\}$.
\end{itemize}
\end{mylem}
\begin{proof}
If $\hat{x}\in \X^z(W)$, there exists $y'\in\Y^z(W)$ satisfying $\hat{p}_{y',z}(\hat{x})\neq 0$. Fix $y\in \Y^z(W)$ and let $a>0$ be the order of $y-y'$ in $G_2$ (i.e., $a(y-y')=0$). Let $n>0$ be such that $a<2^n$ and define the two sequences $(y_i)_{1\leq i\leq2^n}$ and $(y_i')_{1\leq i\leq2^n}$ as follows:
\begin{itemize}
\item If $1\leq i\leq a$, $y_i=y$ and $y_i'=y'$.
\item If $a<i\leq 2^n$, $y_i=y_i'=y'$.
\end{itemize}
Since $a(y-y')=0$, we have $ay=ay'$ and so $\displaystyle\sum_{i=1}^{2^n} y_i=ay+(2^n-a)y'=ay'+(2^n-a)y'=\sum_{i=1}^{2^n} y_i'$. By applying Lemma \ref{lem2}, we get
\begin{align*}
\big(\hat{p}_{y,z}(\hat{x})\big)^a\big(\hat{p}_{y',z}(\hat{x})\big)^{2^n-a}&=\prod_{i=1}^{2^n} \hat{p}_{y_i,z}(\hat{x})=\prod_{i=1}^{2^n} \hat{p}_{y_i',z}(\hat{x})\\
&=\big(\hat{p}_{y',z}(\hat{x})\big)^{2^n}\neq 0.
\end{align*}
Therefore, $\hat{p}_{y,z}(\hat{x})\neq 0$. Moreover, $$\Big(\frac{\hat{p}_{y,z}(\hat{x})}{\hat{p}_{y',z}(\hat{x})}\Big)^a=1,$$
which means that $\displaystyle\frac{\hat{p}_{y,z}(\hat{x})}{\hat{p}_{y',z}(\hat{x})}$ is a root of unity. Hence $\displaystyle\frac{\hat{p}_{y,z}(\hat{x})}{\hat{p}_{y',z}(\hat{x})}\in\mathbb{T}$.
\end{proof}

\begin{mylem}
If $I_1$ is $\ast^-$ preserved for $W$, there exists a unique mapping $\hat{f}_W:\XDY(W)\rightarrow\mathbb{T}$ such that for every $(\hat{x},z)\in\XZ(W)$ and every $y_1,y_2\in\Y^z(W)$, we have $$\hat{p}_{y_1,z}(\hat{x})=\hat{f}_W (\hat{x},y_1-y_2)\cdot\hat{p}_{y_2,z}(\hat{x}).$$
\label{lemWminusRec}
\end{mylem}
\begin{proof}
Let $(\hat{x},y)\in \XDY(W)$. Let $z$ be such that $(\hat{x},y)\in\XDY^z(W)=\X^z(W)\times \DY^z(W)$, and let $y_1,y_2\in\Y^{z}(W)$ be such that $y_1-y_2=y$. We want to show that $\displaystyle\frac{\hat{p}_{y_1,z}(\hat{x})}{\hat{p}_{y_2,z}(\hat{x})}$ depends only on $(\hat{x},y)=(\hat{x},y_1-y_2)$ and that $\displaystyle\frac{\hat{p}_{y_1,z}(\hat{x})}{\hat{p}_{y_2,z}(\hat{x})}\in\mathbb{T}$.

Suppose there exist $z'\in\mathcal{Z}$ and $y_1',y_2'\in\Y^{z'}(W)$ which satisfy $\hat{x}\in\X^{z'}(W)$ and $y_1'-y_2'=y=y_1-y_2$. We need to show that $\displaystyle\frac{\hat{p}_{y_1,z}(\hat{x})}{\hat{p}_{y_2,z}(\hat{x})}=\frac{\hat{p}_{y_1',z'}(\hat{x})}{\hat{p}_{y_2',z'}(\hat{x})}\in\mathbb{T}$.

From Lemma \ref{lem3} we have $p_{y_1,z}(\hat{x})\neq 0$, $p_{y_2,z}(\hat{x})\neq 0$, $p_{y_1',z'}(\hat{x})\neq 0$ and $p_{y_2',z'}(\hat{x})\neq 0$. On the other hand, since $y_1+y_2'=y_2+y_1'$, Lemma \ref{lem2} shows that $p_{y_1,z}(\hat{x})\cdot p_{y_2',z'}(\hat{x})=p_{y_2,z}(\hat{x})\cdot p_{y_1',z'}(\hat{x})$.
Therefore, $$\displaystyle\frac{p_{y_1,z}(\hat{x})}{p_{y_2,z}(\hat{x})}=\frac{p_{y_1',z'}(\hat{x})}{p_{y_2',z'}(\hat{x})}\stackrel{(a)}{\in}\mathbb{T},$$ where (a) follows from Lemma \ref{lem3}. This shows that the value of $\displaystyle\frac{p_{y_1,z}(\hat{x})}{p_{y_2,z}(\hat{x})}\in\mathbb{T}$ depends only on $(\hat{x},y)$ and does not depend on the choice of $z,y_1,y_2$. We conclude that there exists a unique $\hat{f}_{W}(\hat{x},y)\in\mathbb{T}$ such that for every $z\in\mathcal{Z}$ and every $y_1,y_2\in\Y^z(W)$ satisfying $\hat{x}\in\X^z(W)$ and $y_1-y_2=y$, we have $\hat{p}_{y_1,z}(\hat{x})=\hat{f}_W(\hat{x},y)\cdot\hat{p}_{y_2,z}(\hat{x})$.
\end{proof}

\vspace*{3mm}

Notice that the only difference between the mapping $\hat{f}_W$ in Lemma \ref{lemWminusRec} and the function $F$ in Definition \ref{defComp} is that $F$ is a pseudo-quadratic function defined on a pseudo-quadratic domain $D$ containing $\XDY(W)$, whereas $\hat{f}_W$ is only defined on $\XDY(W)$. Therefore, if we want to prove that $W$ is polarization compatible, we have to show that $\hat{f}_W$ can be extended to a pseudo-quadratic function.

Another important remark is that if polarization $\ast$-preserves $I_1$ for $W$, then from Lemma \ref{lemRecPres} we can see that $I_1$ is preserved for $W^{(s,[n]^-)}$ for every $s\in\{-,+\}^{\ast}$ and every $n\geq 0$. Therefore, $I_1$ is $\ast^-$ preserved for $W^s$ for every $s\in\{-,+\}^{\ast}$. Lemma \ref{lemWminusRec} now implies that for every  $s\in\{-,+\}^{\ast}$, there exists a function $\hat{f}_{W^s}:\XDY(W^s)\to\mathbb{T}$ such that for every $(\hat{x},z^s)\in\XZ(W^s)$ and every $y_1,y_2\in\Y^{z^s}(W^s)$, we have $$\hat{p}_{y_1,z^s,W^s}(\hat{x})=\hat{f}_{W^s} (\hat{x},y_1-y_2)\cdot\hat{p}_{y_2,z^s,W^s}(\hat{x}).$$
By studying the relations between $\XDY(W)$ and $\hat{f}_W$ on one hand and $\XDY(W^s)$ and $\hat{f}_{W^s}$ on the other hand, we can deduce restrictions on $\hat{f}_W$ which will allow us to extend it to a pseudo-quadratic function.

The following proposition shows how $\XDY(W^-)$ and $\hat{f}_{W^-}$ are related to $\XDY(W)$ and $\hat{f}_W$ in the case where $I_1$ is $\ast^-$ preserved for $W$.

\begin{myprop}
\label{propWminus}
If $I_1$ is $\ast^-$ preserved for $W$, we have:
\begin{enumerate}
\item $\XDY(W^-)=\{(\hat{x},y_1+y_2):\;(\hat{x},y_1),(\hat{x},y_2)\in\XDY(W)\}$.
\item For every $\hat{x}\in G_1$ and every $y_1,y_2\in G_2$ satisfying $(\hat{x},y_1),(\hat{x},y_2)\in\XDY(W)$, we have $$\hat{f}_{W^-}(\hat{x},y_1+y_2)=\hat{f}_W(\hat{x},y_1)\cdot\hat{f}_W(\hat{x},y_2).$$
\end{enumerate}
\end{myprop}
\begin{proof}
See Appendix \ref{apppropWminus}.
\end{proof}

\begin{mycor}
\label{corMinus}
If $I_1$ is $\ast^-$ preserved for $W$, then $\XDY(W)\subset\XDY(W^-)$ and $\hat{f}_{W^-}(\hat{x},y)=\hat{f}_{W}(\hat{x},y)$ for every $(\hat{x},y)\in\XDY(W)$, i.e., $\hat{f}_{W^-}$ is an extension of $\hat{f}_W$.
\end{mycor}
\begin{proof}
Let $(\hat{x},y)\in\XDY(W)$. There exists $z\in\mathcal{Z}$ and $y_1,y_2\in\Y^z(W)$ such that $y=y_1-y_2$, $\hat{p}_{y_1,z}(\hat{x})\neq0$ and $\hat{p}_{y_2,z}(\hat{x})\neq0$. Since $y_1\in\Y^z(W)$, we have $0=y_1-y_1\in\DY^z(W)$. Therefore, we have $(\hat{x},0)\in\XDY(W)$ and $\displaystyle\hat{f}_W(\hat{x},0)=\frac{\hat{p}_{y_1,z}(\hat{x})}{\hat{p}_{y_1,z}(\hat{x})}=1$.

Since $(\hat{x},y)\in\XDY(W)$ and $(\hat{x},0)\in\XDY(W)$, Proposition \ref{propWminus} implies that $(\hat{x},y)=(\hat{x},y+0)\in\XDY(W^-)$ and $\hat{f}_{W^-}(\hat{x},y)=\hat{f}_{W}(\hat{x},y)\cdot\hat{f}_{W}(\hat{x},0)=\hat{f}_{W}(\hat{x},y)$.
\end{proof}

\vspace*{3mm}

The following proposition shows how $\XDY(W^+)$ and $\hat{f}_{W^+}$ are related to $\XDY(W)$ and $\hat{f}_W$ in the case where polarization $\ast$-preserves $I_1$ for $W$.

\begin{myprop}
\label{propWplus}
If polarization $\ast$-preserves $I_1$ for $W$, we have:
\begin{enumerate}
\item $\big\{(\hat{x}_1+\hat{x}_2,y):\;(\hat{x}_1,y),(\hat{x}_2,y)\in\XDY(W)\big\}\subset\XDY(W^+)$.
\item For every $\hat{x}_1,\hat{x}_2\in G_1$ and every $y\in G_2$ satisfying $(\hat{x}_1,y),(\hat{x}_2,y)\in\XDY(W)$, we have $$\hat{f}_{W^+}(\hat{x}_1+\hat{x}_2,y)=\hat{f}_W(\hat{x}_1,y)\cdot\hat{f}_W(\hat{x}_2,y).$$
\end{enumerate}
\end{myprop}
\begin{proof}
See Appendix \ref{apppropWplus}.
\end{proof}

\begin{mycor}
\label{corPlus}
If polarization $\ast$-preserves $I_1$ for $W$, then $\XDY(W)\subset\XDY(W^+)$ and $\hat{f}_{W^+}(\hat{x},y)=\hat{f}_{W}(\hat{x},y)$ for every $(\hat{x},y)\in\XDY(W)$, i.e., $\hat{f}_{W^+}$ is an extension of $\hat{f}_W$.
\end{mycor}
\begin{proof}
For every $(\hat{x},y)\in\XDY(W)$, there exists $z\in\mathcal{Z}$ and $y_1,y_2\in\Y^z(W)$ such that $\hat{x}\in\X^z(W)$ and $y=y_1-y_2$. Lemma \ref{lem3} implies that $\hat{p}_{y_1,z}(\hat{x})\neq0$ and $\hat{p}_{y_2,z}(\hat{x})\neq0$. We have:
$$\hat{p}_{y_1,z}(0)=\sum_{x\in G_1}p_{y_1,z}(x)e^{-j2\pi\langle0,x\rangle}=\sum_{x\in G_1}p_{y_1,z}(x)=1\neq0.$$
Similarly, $\hat{p}_{y_2,z}(0)=1\neq0$. Therefore, we have $0\in\X^z(W)$ and $y\in\DY^z(W)$. Hence, $(0,y)\in\XDY(W)$ and $\displaystyle \hat{f}_W(0,y)=\frac{\hat{p}_{y_1,z}(0)}{\hat{p}_{y_2,z}(0)}=1$.

Since $(\hat{x},y)\in\XDY(W)$ and $(0,y)\in\XDY(W)$, Proposition \ref{propWplus} implies that $(\hat{x},y)=(\hat{x}+0,y)\in\XDY(W^+)$ and $\hat{f}_{W^+}(\hat{x},y)=\hat{f}_{W}(\hat{x},y)\hat{f}_{W}(0,y)=\hat{f}_{W}(\hat{x},y)$.
\end{proof}

\vspace{5mm}

The next proposition gives a necessary condition for the $\ast$-preservation of $I_1$:

\begin{myprop}
\label{PropNec}
If polarization $\ast$-preserves $I_1$ for $W$, then $\hat{f}_W$ can be extended to a pseudo-quadratic function.
\end{myprop}
\begin{proof}
Define the sequence $(W_n)_{n\geq0}$ of MACs recursively as follows:
\begin{itemize}
\item $W_0=W$.
\item $W_n=W_{n-1}^-$ if $n>0$ is odd.
\item $W_n=W_{n-1}^+$ if $n>0$ is even.
\end{itemize}
For example, we have $W_1=W^-$, $W_2=W^{(-,+)}$, $W_3=W^{(-,+,-)}$, $W_4=W^{(-,+,-,+)}$ \ldots

It follows from Corollaries \ref{corMinus} and \ref{corPlus} that:
\begin{itemize}
\item The sequence of sets $\big(\XDY(W_n)\big)_{n\geq0}$ is increasing.
\item $\hat{f}_{W_n}$ is an extension of $\hat{f}_W$ for every $n>0$.
\end{itemize}
Since $\big(\XDY(W_n)\big)_{n\geq0}$ is increasing and since $G_1\times G_2$ is finite, there exists $n_0>0$ such that for every $n\geq n_0$ we have $\XDY(W_n)=\XDY(W_{n_0})$ for all $n\geq n_0$. We may assume without loss of generality that $n_0$ is even. Define the following sets:
\begin{itemize}
\item $\hat{H}_1=\{\hat{x}\in G_1:\;\exists y\in G_2,\;(\hat{x},y)\in\XDY(W_{n_0})\}$.
\item For every $\hat{x}\in \hat{H}_1$, let $H_2^{\hat{x}}=\{y\in G_2:\;(\hat{x},y)\in\XDY(W_{n_0})\}$.
\item $H_2=\{y\in G_2:\;\exists \hat{x}\in G_1,\;(\hat{x},y)\in\XDY(W_{n_0})\}$.
\item For every $y\in H_2$, let $\hat{H}_1^{y}=\{\hat{x}\in G_1:\;(\hat{x},y)\in\XDY(W_{n_0})\}$.
\end{itemize}
We have the following:
\begin{itemize}
\item For every fixed $y\in H_2$, let $\hat{x}_1,\hat{x}_2\in \hat{H}_1^{y}$ so that $(\hat{x}_1,y),(\hat{x}_2,y)\in\XDY(W_{n_0})\subset \XDY(W_{n_0+1})$. It follows from Proposition \ref{propWplus} that $(\hat{x}_1+\hat{x}_2,y)\in\XDY(W_{n_0+1}^+)=\XDY(W_{n_0+2})=\XDY(W_{n_0})$ which implies that $\hat{x}_1+\hat{x}_2\in \hat{H}_1^y$. Hence $\hat{H}_1^y$ is a subgroup of $(G_1,+)$. Moreover, we have:
\begin{align*}
\hat{f}_{W_{n_0}}(\hat{x}_1+\hat{x}_2,y)&\stackrel{(a)}{=}\hat{f}_{W_{n_0+2}}(\hat{x}_1+\hat{x}_2,y)=\hat{f}_{W_{n_0+1}^+}(\hat{x}_1+\hat{x}_2,y)\\
&\stackrel{(b)}{=}\hat{f}_{W_{n_0+1}}(\hat{x}_1,y)\cdot\hat{f}_{W_{n_0+1}}(\hat{x}_2,y)\stackrel{(c)}{=}\hat{f}_{W_{n_0}}(\hat{x}_1,y)\cdot\hat{f}_{W_{n_0}}(\hat{x}_2,y),
\end{align*}
where (a) and (c) follow from corollaries \ref{corMinus} and \ref{corPlus} and (b) follows from Proposition \ref{propWplus}. Therefore the mapping $\hat{x}\rightarrow \hat{f}_{W_{n_0}}(\hat{x},y)$ is a group homomorphism from $(\hat{H}_1^y,+)$ to $(\mathbb{T},\cdot)$.
\item For every fixed $\hat{x}\in \hat{H}_1$, let $y_1,y_2\in H_2^{\hat{x}}$ so that $(\hat{x},y_1),(\hat{x},y_2)\in\XDY(W_{n_0})$. It follows from Proposition \ref{propWminus} that $(\hat{x},y_1+y_2)\in\XDY(W_{n_0}^-)=\XDY(W_{n_0+1})=\XDY(W_{n_0})$ which implies that $y_1+y_2\in H_2^{\hat{x}}$. Hence $H_2^{\hat{x}}$ is a subgroup of $(G_2,+)$. Moreover, we have
\begin{align*}
\hat{f}_{W_{n_0}}(\hat{x},y_1+y_2)&\stackrel{(a)}{=}\hat{f}_{W_{n_0+1}}(\hat{x},y_1+y_2)=\hat{f}_{W_{n_0}^-}(\hat{x},y_1+y_2)\\
&\stackrel{(b)}{=}\hat{f}_{W_{n_0}}(\hat{x},y_1)\cdot\hat{f}_{W_{n_0}}(\hat{x},y_2),
\end{align*}
where (a) follows from corollary \ref{corMinus} and (b) follows from Proposition \ref{propWminus}. Therefore the mapping $y\rightarrow \hat{f}_{W_{n_0}}(\hat{x},y)$ is a group homomorphism from $(H_2^{\hat{x}},+)$ to $(\mathbb{T},\cdot)$.
\end{itemize}
We conclude that $\hat{f}_{W_{n_0}}$ (which is an extension of $\hat{f}_W$) is pseudo-quadratic.
\end{proof}

\vspace{5mm}

Proposition \ref{PropNec} shows that if polarization $\ast$-preserves $I_1$ for $W$ then $W$ must be polarization compatible with respect to the first user.

For the sake of simplicity and brevity, we will write ``polarization compatible" to denote ``polarization compatible with respect to the first user".

\subsection{Polarization compatibility is sufficient}

In this subsection, we show that polarization compatibility is a sufficient condition for the $\ast$-preservation of $I_1$.

\begin{mylem}
\label{lemSuf1}
If $W:G_1\times G_2\longrightarrow\mathcal{Z}$ is polarization compatible then $I_1$ is preserved for $W$.
\end{mylem}
\begin{proof}
Let $F:D\rightarrow \mathbb{T}$ be the pseudo-quadratic function of Definition \ref{defComp}. Suppose that $y_1,y_2,y_1',y_2'\in G_2$ and $z_1,z_2\in\mathcal{Z}$ satisfy:
\begin{itemize}
\item $y_1-y_2=y_1'-y_2'$.
\item $y_1,y_1'\in\Y^{z_1}(W)$ and $y_2,y_2'\in\Y^{z_2}(W)$.
\end{itemize}
For every $\hat{x}\in G_1$, we have:
\begin{itemize}
\item If $(\hat{x},z_1)\notin\XZ(W)$ then $\hat{p}_{y_1,z_1}(\hat{x})=0$ and $\hat{p}_{y_1',z_1}(\hat{x})=0$, so $$\hat{p}_{y_1,z_1}(\hat{x})\hat{p}_{y_2,z_2}(\hat{x})^\ast=\hat{p}_{y_1',z_1}(\hat{x})\hat{p}_{y_2',z_2}(\hat{x})^\ast=0.$$
\item If $(\hat{x},z_2)\notin\XZ(W)$ then $\hat{p}_{y_2,z_2}(\hat{x})=0$ and $\hat{p}_{y_2',z_2}(\hat{x})=0$, so $$\hat{p}_{y_1,z_1}(\hat{x})\hat{p}_{y_2,z_2}(\hat{x})^\ast=\hat{p}_{y_1',z_1}(\hat{x})\hat{p}_{y_2',z_2}(\hat{x})^\ast=0.$$
\item If $(\hat{x},z_1)\in\XZ(W)$ and $(\hat{x},z_2)\in\XZ(W)$, then
\begin{align*}
\hat{p}_{y_1,z_1}(\hat{x})\hat{p}_{y_2,z_2}(\hat{x})^\ast=\hat{p}_{y_1',z_1}(\hat{x})F(\hat{x},y_1-y_1')\hat{p}_{y_2',z_2}(\hat{x})^\ast F(\hat{x},y_2-y_2')^{\ast}\stackrel{(a)}{=}\hat{p}_{y_1',z_1}(\hat{x})\hat{p}_{y_2',z_2}(\hat{x})^\ast,
\end{align*}
where (a) follows from the fact that $y_1-y_1'=y_2-y_2'$ and so $F(\hat{x},y_1-y_1')F(\hat{x},y_2-y_2')^\ast=|F(\hat{x},y_1-y_1')|^2=1$.
\end{itemize}
Therefore, we have $\hat{p}_{y_1,z_1}(\hat{x})\hat{p}_{y_2,z_2}(\hat{x})^\ast=\hat{p}_{y_1',z_1}(\hat{x})\hat{p}_{y_2',z_2}(\hat{x})^\ast$ for all $\hat{x}\in G_1$. Lemma \ref{lemLemlemLem} now implies that $I_1$ is preserved for $W$.
\end{proof}

\begin{mylem}
\label{lemSufMinusPlus}
If $W:G_1\times G_2\longrightarrow\mathcal{Z}$ is polarization compatible then $W^-$ and $W^+$ are polarization compatible as well.
\end{mylem}
\begin{proof}
See Appendix \ref{applemSufMinusPlus}.
\end{proof}

\begin{myprop}
\label{propSuf}
If $W$ is polarization compatible then polarization $\ast$-preserves $I_1$ for $W$.
\end{myprop}
\begin{proof}
Suppose that $W$ is polarization compatible. Using Lemma \ref{lemSufMinusPlus}, we can show by induction that $W^s$ is polarization compatible for every $s\in\{-,+\}^\ast$. Lemma \ref{lemSuf1} now implies that $I_1$ is preserved for $W^s$ for every $s\in\{-,+\}^\ast$. By applying Lemma \ref{lemRecPres}, we deduce that polarization $\ast$-preserves $I_1$ for $W$.
\end{proof}

\vspace*{3mm}

Propositions \ref{PropNec} and \ref{propSuf} show that polarization $\ast$-preserves $I_1$ for $W$ if and only if $W$ is polarization compatible. This completes the proof of Theorem \ref{MainThe}.

\subsection{Special cases}
The characterization found in Theorem \ref{MainThe} (i.e., polarization compatibility) takes a simple form in the special case where $G_1=G_2=\mathbb{F}_q$ for a prime $q$:

\begin{myprop}
Let $W:\mathbb{F}_q\times \mathbb{F}_q\longrightarrow\mathcal{Z}$ be a two-user MAC and let $(X,Y)\stackrel{W}{\longrightarrow}Z$. Polarization $\ast$-preserves $I_1$ for $W$ if and only if there exists $a\in\mathbb{F}_q$ such that $I(X+aY;Y|Z)=0$.
\label{propSimple}
\end{myprop}
\begin{proof}
If polarization $\ast$-preserves $I_1$ for $W$ then $W$ is polarization compatible. Let $F:D\rightarrow\mathbb{T}$ be the pseudo-quadratic function of Definition \ref{defComp}. We have the following:
\begin{itemize}
\item If there exists $(\hat{x},y)\in D$ such that $\hat{x}\neq0$ and $y\neq0$ then $D=\mathbb{F}_q\times\mathbb{F}_q$ since $D$ is a pseudo-quadratic domain and since $q$ is prime.
\item If for all $(\hat{x},y)\in D$ we have either $\hat{x}=0$ or $y=0$, then $F(\hat{x},y)=1$ for every $(\hat{x},y)\in D$. Hence the mapping $F':\mathbb{F}_q\times\mathbb{F}_q\rightarrow\mathbb{T}$ defined as $F'(\hat{x},y)=1$ is an extension of $F$ which preserves the pseudo-quadratic property.
\end{itemize}
Therefore, we can assume without loss of generality that $D=\mathbb{F}_q\times\mathbb{F}_q$. Now since $F(1,1)^q=F(1,q\cdot 1)=F(1,0)=1$, $F(1,1)$ is a $q^{th}$ root of unity. Therefore, there exists $a\in\mathbb{F}_q$ such that $F(1,1)=e^{j2\pi\frac{a}{q}}$.

Fix $z\in\mathcal{Z}$ and $y_1,y_2\in\Y^z(W)$. For every $\hat{x}\in\mathbb{F}_q$ we have $$\hat{p}_{y_1,z}(\hat{x})=\hat{p}_{y_2,z}(\hat{x})\cdot F(\hat{x},y_1-y_2)=\hat{p}_{y_2,z}(\hat{x})\cdot e^{j2\pi a\frac{(y_1-y_2)\hat{x}}{q}},$$ which is equivalent to say that for every $x'\in\mathbb{F}_q$, we have $p_{y_1,z}(x')=p_{y_2,z}(x'+a(y_1-y_2))$, i.e.,
\begin{equation}
\label{sdnvbshgdfse}
P_{X|Y,Z}(x'|y_1,z)=P_{X|Y,Z}(x'+a(y_1-y_2)|y_2,z).
\end{equation}
By applying the change of variable $x'=x-ay_1$, we can see that \eqref{sdnvbshgdfse} is equivalent to
\begin{align*}
P_{X+aY|Y,Z}(x|y_1,z)&=P_{X|Y,Z}(x-ay_1|y_1,z)=P_{X|Y,Z}(x'|y_1,z)\\
&=P_{X|Y,Z}(x'+a(y_1-y_2)|y_2,z)=P_{X|Y,Z}(x-ay_1+a(y_1-y_2)|y_2,z)\\
&=P_{X|Y,Z}(x-ay_2|y_2,z)=P_{X+aY|Y,Z}(x|y_2,z).
\end{align*}
This shows that $X+aY$ is conditionally independent of $Y$ given $Z$, i.e., $I(X+aY;Y|Z)=0$.

On the other hand, let $W:\mathbb{F}_q\times\mathbb{F}_q\longrightarrow\mathcal{Z}$ be a two-user MAC and let $(X,Y)\stackrel{W}{\longrightarrow}Z$. If there exists $a\in\mathbb{F}_q$ such that $I(X+aY;Y|Z)=0$, then Proposition \ref{propSufFirst} implies that polarization $\ast$-preserves $I_1$ for $W$.
\end{proof}

\begin{mycor}
Polarization $\ast$-preserves the symmetric capacity region for the binary adder channel.
\label{corBAC}
\end{mycor}
\begin{proof}
Let $X$ and $Y$ be two independent uniform random variables in $\{0,1\}$, and let $Z=X+Y\in\{0,1,2\}$ (where $+$ here denotes addition in $\mathbb{R}$). It is easy to check that $I(X\oplus Y;Y|Z)=I(X\oplus Y;X|Z)=0$. Therefore, polarization $\ast$-preserves $I_1$ and $I_2$ for $W$. We conclude that polarization $\ast$-preserves the symmetric capacity region for $W$.
\end{proof}

\begin{myrem}
It may seem promising to try to generalize Proposition \ref{propSimple} to the case where $G_1=\mathbb{F}_q^k$ and $G_2=\mathbb{F}_q^l$ by considering the condition $I(X+AY;Y|Z)=0$ for some matrix $A\in\mathbb{F}_q^{k\times l}$. Although this condition is sufficient for the $\ast$-preservation of $I_1$ (Proposition \ref{propSufFirst}), it turns out that it is not necessary.
\end{myrem}

\begin{myprop}
\label{propDif}
If $|G_1|$ and $|G_2|$ are co-prime and $(X,Y)\stackrel{W}{\longrightarrow}Z$, then polarization $\ast$-preserves $I_1$ for $W$ if and only if $I(X;Y|Z)=0$ (i.e., if and only if the dominant face of $\mathcal{J}(W)$ is a single point).
\end{myprop}
\begin{proof}
Let $F:D\rightarrow\mathbb{T}$ be a pseudo-quadratic function. For every $(\hat{x},y)\in D$, we have:
\begin{itemize}
\item $F(\hat{x},y)^{|G_1|}=F(|G_1|\cdot \hat{x},y)=F(0,y)=1$.
\item $F(\hat{x},y)^{|G_2|}=F(\hat{x},|G_2|\cdot y)=F(\hat{x},0)=1$.
\end{itemize}
Therefore, $F(\hat{x},y)$ is both a $|G_1|^{th}$ root of unity and a $|G_2|^{th}$ root of unity. This shows that $F(\hat{x},y)$ must be equal to 1 because $|G_1|$ and $|G_2|$ are co-prime. We conclude that every pseudo-quadratic function $F:D\rightarrow\mathbb{T}$ must be equal to 1. Therefore, polarization $\ast$-preserves $I_1$ for $W$ if and only if $\hat{p}_{y_1,z}(\hat{x})=\hat{p}_{y_2,z}(\hat{x})$ for every $(\hat{x},z)\in\XZ(W)$ and every $y_1,y_2\in\Y^z(W)$.

Now since $\hat{p}_{y_1,z}(\hat{x})=\hat{p}_{y_2,z}(\hat{x})=0$ for every $\hat{x}\notin\X^z(W)$ and every $y_1,y_2\in\Y^z(W)$, we conclude that polarization $\ast$-preserves $I_1$ for $W$ if and only if $\hat{p}_{y_1,z}(\hat{x})=\hat{p}_{y_2,z}(\hat{x})$ for every $(\hat{x},z)\in G_1\times\mathcal{Z}$ and every $y_1,y_2\in\Y^z(W)$. This is equivalent to say that $p_{y_1,z}(x)=p_{y_2,z}(x)$ for every $(x,z)\in G_1\times\mathcal{Z}$ and every $y_1,y_2\in\Y^z(W)$. This just means that $X$ and $Y$ are conditionally independent given $Z$.
\end{proof}

\section{Generalization to multiple user MACs}

\begin{mydef}
Let $W:G_1\times\ldots\times G_m\longrightarrow \mathcal{Z}$ be an $m$-user MAC. For every $S\subset\{1,\ldots,m\}$, we define the two-user MAC $W_S:G_{S}\times G_{S^c}\longrightarrow\mathcal{Z}$ as $W_S(y|x_S,x_{S^c})=W(y|x_1,\ldots,x_m)$.
\end{mydef}

\begin{myrem}
It is easy to see that for every $s\in\{-,+\}^\ast$ and every $S\subset\{1,\ldots,m\}$, we have $(W^s)_S=(W_S)^s$. Therefore, polarization $\ast$-preserves $I_S$ for $W$ if and only if polarization $\ast$-preserves $I_1$ for $W_S$.
\label{Remsdfsjegf}
\end{myrem}

\begin{mythe}
\label{MainThe2} Let $W:G_1\times\ldots\times G_m\longrightarrow \mathcal{Z}$ be an $m$-user MAC. Polarization $\ast$-preserves $I_S$ for $W$ if and only if $W_S$ is polarization compatible.
\end{mythe}
\begin{proof}
Direct corollary of Theorem \ref{MainThe} and Remark \ref{Remsdfsjegf}.
\end{proof}

\section{Discussion and Conclusion}

The necessary and sufficient condition that we provided is a single letter characterization: the mapping $\hat{f}_W$ can be directly computed using the transition probabilities of $W$. Moreover, since the number of pseudo-quadratic functions is finite, checking whether $\hat{f}_W$ is extendable to a pseudo-quadratic function can be accomplished in a finite number of computations.

\appendices

\section{Proof of Proposition \ref{propCharacInd}}

\label{appCharacInd}

We need the following lemma:

\begin{mylem}
\label{lemGroupHomoCharac}
Let $(G,+)$ be an Abelian group and let $\hat{f}:G\to\mathbb{T}$ be a group homomorphism from $(G,+)$ to $(\mathbb{T},\cdot)$. There exists $x_f\in G$ satisfying:
\begin{itemize}
\item $\hat{f}(\hat{x})=e^{j2\pi\langle \hat{x},x_f\rangle}$ for every $\hat{x}\in G$.
\item If $n>0$ is such that $\hat{f}(\hat{x})^n=1$ for every $\hat{x}\in G$, then $nx_f=0$.
\end{itemize}
\end{mylem}
\begin{proof}
Let $N_1,\ldots,N_k>0$ be $k$ integers such that $G=\mathbb{Z}_{N_1}\times\ldots\times \mathbb{Z}_{N_k}$. For every $1\leq i\leq k$, let $e_i=(0,\ldots,0,1,0,\ldots,0)\in G$ be the element of $G$ whose $j^{th}$ coordinate is $1$ if $j=i$ and $0$ otherwise.

Since $N_i e_i=0$, we have $\hat{f}(e_i)^{N_i}=\hat{f}(N_ie_i)=\hat{f}(0)=1$. Therefore, $\hat{f}(e_i)$ is an $N_i^{th}$ root of unity, so there exists $0\leq x_i< N_i$ such that $\hat{f}(e_i)=e^{\frac{j2\pi x_i}{N_i}}$.

Let $x_f:=(x_1,\ldots,x_k)\in G$. For every $\hat{x}=(\hat{x}_1,\ldots\hat{x}_k)\in G$ we have:
\begin{align*}
\hat{f}(\hat{x})=\hat{f}\left(\sum_{i=1}^k\hat{x}_ie_i\right)=\prod_{i=1}^k \hat{f}(e_i)^{\hat{x}_i}=\prod_{i=1}^k \left(e^{\frac{j2\pi x_i}{N_i}}\right)^{\hat{x}_i}=e^{\sum_{i=1}^k \frac{j2\pi \hat{x}_ix_i}{N_i}}=e^{j2\pi\langle\hat{x},x_f\rangle}.
\end{align*}

If $n>0$ is such that $\hat{f}(\hat{x})^n=1$ for every $\hat{x}\in G$, then $e^{\frac{j2\pi n x_i}{N_i}}=  \hat{f}(e_i)^n=1$ for every $1\leq i\leq k$. This means that $N_i$ divides $n x_i$ for every $1\leq i\leq k$. Therefore,
$$nx_f=(nx_1 \bmod N_1,\ldots,nx_k \bmod N_k)=0.$$
\end{proof}

\vspace*{3mm}

Now we are ready to prove Proposition \ref{propCharacInd}.

Since we have shown the necessary condition in the discussion before the statement of Proposition \ref{propCharacInd}, we only need to show the sufficient condition.

Let $W:G_1\times G_2\longrightarrow \mathcal{Z}$ be a two-user MAC, and assume that there exists a subgroup $H_2$ of $G_2$ and a pseudo-quadratic function $F:G_1\times H_2\to\mathbb{T}$ satisfying:
\begin{itemize}
\item $\XDY(W)\subset G_1\times H_2$.
\item For every $(\hat{x},z)\in \XZ(W)$ and every $y_1,y_2\in\Y^z(W)$, we have $\hat{p}_{y_1,z}(\hat{x})=F(\hat{x},y_1-y_2)\hat{p}_{y_2,z}(\hat{x})$.
\end{itemize}

Since $(H_2,+)$ is an Abelian group, it is isomorphic to the product of cyclic groups. Let $N_1',\ldots,N_{k'}'>0$ be $k'$ integers such that $H_2$ is isomorphic to $\mathbb{Z}_{N_1'}\times\ldots\times\mathbb{Z}_{N_{k'}'}$. Because of this isomorphism, we can find $k'$ elements $e_1',\ldots,e_{k'}'\in H_2$ such that:
\begin{itemize}
\item $e_i'$ is of order $N_i'$ for every $1\leq i\leq k'$.
\item For every $y\in H_2$, there exist unique integers $0\leq y_1< N_1'$, \ldots, $0\leq y_{k'}< N_{k'}'$ such that $\displaystyle y=\sum_{i=1}^{k'}y_i e_i'$.
\end{itemize}

For every $1\leq i\leq k'$, the mapping $\hat{x}\to F(\hat{x},e_i')$ is a group homomorphism from $(G_1,+)$ to $(\mathbb{T},\cdot)$. Lemma \ref{lemGroupHomoCharac} shows that there exists $f_i\in G_1$ such that $F(\hat{x},e_i')=e^{j2\pi\langle\hat{x},f_i\rangle}$ for every $\hat{x}\in G_1$. Moreover, for every $1\leq i\leq k'$, we have $F(\hat{x},e_i')^{N_i'}=F(\hat{x},N_i'e_i')=F(\hat{x},0)=1$ for every $\hat{x}\in G_1$, hence $N_i'f_i=0$.

For every $y\in H_2$, define $\displaystyle f(y)=\sum_{i=1}^{k'}y_if_i$, where $0\leq y_1< N_1'$, \ldots, $0\leq y_{k'}< N_{k'}'$ satisfy $\displaystyle y=\sum_{i=1}^{k'}y_i e_i'$. We can show that $f$ is a group homomorphism from $(H_2,+)$ to $(G_1,+)$: Let $y,y'\in H_2$ and let $0\leq y_1,y_1',y_1''\leq N_1'$,  \ldots, $0\leq y_{k'},y_{k'}',y_{k'}''\leq N_{k'}'$ be such that $\displaystyle y=\sum_{i=1}^{k'}y_i e_i'$, $\displaystyle y'=\sum_{i=1}^{k'}y_i' e_i'$ and $\displaystyle y+y'=\sum_{i=1}^{k'}y_i'' e_i'$. We have:
$$0=y+y'-y-y'=\sum_{i=1}^{k'} (y_i''-y_i-y_i')e_i'\stackrel{(a)}{=}\sum_{i=1}^{k'} (y_i''-y_i-y_i' \bmod N_i')e_i',$$
where (a) follows from the fact that $e_i'$ is of order $N_i'$. Thus, $y_i''=y_i+y_i'\bmod N_i'$  for every $1\leq i\leq k'$. Therefore,
$$f(y+y')=\sum_{i=1}^{k'}y_i'' f_i\stackrel{(b)}{=} \sum_{i=1}^{k'}(y_i+y_i') f_i=\left(\sum_{i=1}^{k'}y_if_i\right)+\left(\sum_{i=1}^{k'}y_i'f_i\right)=f(y)+f(y'),$$
where (b) follows from the fact that $N_i'f_i=0$ and $y_i''=y_i+y_i'\bmod N_i'$  for every $1\leq i\leq k'$. We conclude that $f$ is a group homomorphism from $(H_2,+)$ to $(G_1,+)$.

On the other hand, for every $\hat{x}\in G_1$, we have:
\begin{align*}
F(\hat{x},y)&=F\left(\hat{x},\sum_{i=1}^{k'}y_ie_i'\right)=\prod_{i=1}^{k'} F\left(\hat{x},e_i'\right)^{y_i}=\prod_{i=1}^{k'}\left(e^{j2\pi\langle\hat{x},f_i\rangle}\right)^{y_i}\\
&=e^{\sum_{i=1}^{k'} j2\pi y_i\langle\hat{x},f_i\rangle}=e^{j2\pi\left\langle\hat{x},\sum_{i=1}^{k'} y_if_i\right\rangle}=e^{j2\pi\langle\hat{x},f(y)\rangle}.
\end{align*}

Let $z\in\mathcal{Z}$ and $y_1,y_2\in\Y^z(W)$. For every $\hat{x}\in G_1$, we have:
\begin{itemize}
\item If $\hat{x}\notin\X^z(W)$, we have $\hat{p}_{y_1,z}(\hat{x})=\hat{p}_{y_2,z}(\hat{x})=0$, hence $\hat{p}_{y_1,z}(\hat{x})=\hat{p}_{y_2,z}(\hat{x})e^{j2\pi\langle\hat{x},f(y_1-y_2)\rangle}$.
\item If $\hat{x}\in\X^z(W)$, we have $\hat{p}_{y_1,z}(\hat{x})=\hat{p}_{y_2,z}(\hat{x})F(\hat{x},y_1-y_2)=\hat{p}_{y_2,z}(\hat{x})e^{j2\pi\langle\hat{x},f(y_1-y_2)\rangle}$.
\end{itemize}

We conclude that
\begin{align*}
&\hat{p}_{y_1,z}(\hat{x})=\hat{p}_{y_2,z}(\hat{x})e^{j2\pi\langle\hat{x},f(y_1-y_2)\rangle}\;\; \forall \hat{x}\in G_1\\
\Leftrightarrow\;\;&p_{y_1,z}(x)=p_{y_2,z}(x+f(y_1-y_2))\;\; \forall x\in G_1.
\end{align*}

Lemma \ref{lemCharacInd1} now shows that $W$ is homomorphic-independent.

\section{Proof of Lemma \ref{lem2}}

\label{applem2}

We need the following two lemmas:

\begin{mylem}
\label{lemlem}
Suppose that $I_1$ is $\ast^-$ preserved for $W$. Fix $n>0$ and let $(U_i,V_i)_{0\leq i< 2^n}$ be a sequence of random pairs which are independent and uniformly distributed in $G_1\times G_2$. Let
$$F=\begin{bmatrix}
1 & 1\\
0 & 1
\end{bmatrix}.$$
Define $X_0^{2^n-1}=F^{\otimes n}\cdot U_0^{2^n-1}$ and $Y_0^{2^n-1}=F^{\otimes n}\cdot V_0^{2^n-1}$, and for each $0\leq i<2^n$ let $(X_i,Y_i)\stackrel{W}{\longrightarrow}Z_i$. We have the following:
\begin{itemize}
\item The MAC $(U_0,V_0)\longrightarrow Z_0^{2^n-1}$ is equivalent to $W^{[n]^-}$.
\item $I(U_0;V_1^{2^n-1}|Z_0^{2^n-1}V_0)=0$.
\end{itemize}
\end{mylem}
\begin{proof}
We will show the lemma by induction on $n>0$. For $n=1$, the claim follows from Remark \ref{MainRem} and from the fact that $I_1$ is preserved for $W$ if and only if $I(U_0;V_1|Z_0Z_1V_0)=0$ (see \eqref{mainineq}).

Now let $n>1$ and suppose that the claim is true for $n-1$. Let $N=2^{n-1}$. We have $X_0^{2^n-1}=F^{\otimes n}\cdot U_0^{2^n-1}$ and $Y_0^{2^n-1}=F^{\otimes n}\cdot V_0^{2^n-1}$, i.e., $X_0^{2N-1}=F^{\otimes n}\cdot U_0^{2N-1}$ and $Y_0^{2N-1}=F^{\otimes n}\cdot V_0^{2N-1}$. Therefore, we have:
\vspace*{2mm}
$$X_0^{N-1}=F^{\otimes (n-1)}\cdot(U_0^{N-1}+U_{N}^{2N-1}),$$
$$X_{N}^{2N-1}=F^{\otimes (n-1)}\cdot U_{N}^{2N-1},$$
$$Y_0^{N-1}=F^{\otimes (n-1)}\cdot(V_0^{N-1}+V_{N}^{2N-1}),$$
and
$$Y_{N}^{2N-1}=F^{\otimes (n-1)}\cdot V_{N}^{2N-1}.$$

This means that $(U_0^{N-1}+U_{N}^{2N-1}, V_0^{N-1}+V_{N}^{2N-1},Z_0^{N-1})$ and $(U_{N}^{2N-1}, V_{N}^{2N-1},Z_{N}^{2N-1})$ satisfy the conditions of the induction hypothesis. Therefore,
\begin{itemize}
\item $I(U_0+U_{N};V_1^{N-1}+V_{N+1}^{2N-1}|Z_0^{N-1},V_0+V_{N})=0$.
\item $I(U_{N};V_{N+1}^{2N-1}|Z_{N}^{2N-1},V_{N})=0$.
\end{itemize}
Moreover, since $(U_0^{N-1}+U_{N}^{2N-1}, V_0^{N-1}+V_{N}^{2N-1},Z_0^{N-1})$ is independent of $(U_{N}^{2N-1}, V_{N}^{2N-1},Z_{N}^{2N-1})$, we can combine the above two equations to get:
$$I(U_0+U_{N},U_{N};V_1^{N-1}+V_{N+1}^{2N-1}, V_{N+1}^{2N-1}|Z_0^{2N-1},V_0+V_{N},V_{N})=0,$$
which can be rewritten as
\begin{equation}
\label{eqRecMinusI}
I(U_0U_{N};V_1^{N-1}V_{N+1}^{2N-1}|Z_0^{2N-1}V_0V_{N})=0.
\end{equation}

On the other hand, it also follows from the induction hypothesis that:
\begin{itemize}
\item The MAC $(U_0+U_{N},V_0+V_{N})\longrightarrow Z_0^{N-1}$ is equivalent to $W^{[n-1]^-}$.
\item The MAC $(U_{N},V_{N})\longrightarrow Z_{N}^{2N-1}$ is equivalent to $W^{[n-1]^-}$.
\end{itemize}
This implies that the MAC $(U_0,V_0)\longrightarrow Z_0^{2N-1}$ is equivalent to $W^{[n]^-}$. Now since $I_1$ is $\ast^-$ preserved for $W$, $I_1$ must be preserved for $W^{[n-1]^-}$. Therefore,
\begin{equation}
\label{eqRecMinusII}
I(U_0;V_{N}|Z_0^{2N-1}V_0)=I(U_0;V_{N}|Z_0^{N-1}Z_{N}^{2N-1}V_0)\stackrel{(a)}{=}0,
\end{equation}
where (a) follows from \eqref{mainineq}. We conclude that:
\begin{align*}
I(U_0;V_1^{2N-1}|Z_0^{2N-1}V_0)&=I(U_0;V_{N}|Z_0^{2N-1}V_0)+
I(U_0;V_1^{N-1}V_{N+1}^{2N-1}|Z_0^{2N-1}V_0V_{N})\\
&\leq I(U_0;V_{N}|Z_0^{2N-1}V_0)+
I(U_0U_{N};V_1^{N-1}V_{N+1}^{2N-1}|Z_0^{2N-1}V_0V_{N})\stackrel{(b)}{=}0,
\end{align*}
where (b) follows from \eqref{eqRecMinusI} and \eqref{eqRecMinusII}.
\end{proof}

\begin{mylem}
\label{lemF}
For every $n>0$, if $X_0^{2^n-1}=F^{\otimes n}U_0^{2^n-1}$, then
$\displaystyle U_0=\sum_{i=0}^{2^n-1}(-1)^{|i|_b}X_i$, where $|i|_b$ is the number of ones in the binary expansion of $i$.
\end{mylem}
\begin{proof}
We will show the lemma by induction on $n>0$. For $n=1$, the fact that $X_0^1=F^{\otimes 1}\cdot U_0^1=F\cdot U_0^1$ implies that $X_0=U_0+U_1$ and $X_1=U_1$. Therefore $\displaystyle U_0=X_0-X_1=\sum_{i=0}^1(-1)^{|i|_b}X_i$.

Now let $n>1$ and suppose that the claim is true for $n-1$. Let $N=2^{n-1}$. The fact that $X_0^{2N-1}=F^{\otimes n}\cdot U_0^{2N-1}$ implies that:
\begin{itemize}
\item $X_0^{N-1}=F^{\otimes(n-1)}\cdot(U_0^{N-1}+U_{N}^{2N-1})$.
\item $X_{N}^{2N-1}=F^{\otimes(n-1)}\cdot U_{N}^{2N-1}$.
\end{itemize}
We can apply the induction hypothesis to get:
\begin{itemize}
\item $\displaystyle U_0+U_{N}=\sum_{i=0}^{N-1}(-1)^{|i|_b}X_i$.
\item $\displaystyle U_{N}=\sum_{i=0}^{N-1}(-1)^{|i|_b}X_{i+N}$.
\end{itemize}
Therefore,
\begin{align*}
U_0&=\sum_{i=0}^{N-1}(-1)^{|i|_b}X_i-\sum_{i=0}^{N-1}(-1)^{|i|_b}X_{i+N}=\sum_{i=0}^{N-1}(-1)^{|i|_b}X_i+\sum_{i=0}^{N-1}(-1)^{1+|i|_b}X_{i+N}\\
&=\sum_{i=0}^{N-1}(-1)^{|i|_b}X_i+\sum_{i=N}^{2N-1}(-1)^{1+|i-N|_b}X_i\stackrel{(a)}{=}\sum_{i=0}^{N-1}(-1)^{|i|_b}X_i+\sum_{i=N}^{2N-1}(-1)^{|i|_b}X_i\\
&=\sum_{i=0}^{2N-1}(-1)^{|i|_b}X_i,
\end{align*}
where (a) follows from the fact that for $2^n=N\leq i< 2N=2^{n+1}$, we have $|i-N|_b=|i-2^n|_b=|i|_b-1$.
\end{proof}

\vspace*{3mm}

We are now ready to prove Lemma \ref{lem2}:

Let $W$ be a two-user MAC such that $I_1$ is $\ast^-$ preserved for $W$. Let $n>0$, $y_1,\ldots,y_{2^n},y_1',\ldots, y_{2^n}'\in G_2$ and $z_1,\ldots, z_{2^n}\in\mathcal{Z}$ be such that
\begin{itemize}
\item $\displaystyle \sum_{i=1}^{2^n} y_i=\sum_{i=1}^{2^n} y_i'$,
\item $y_1\in\Y^{z_1}(W),\ldots, y_{2^n}\in\Y^{z_{2^n}}(W)$, and
\item $y_1'\in\Y^{z_1}(W),\ldots, y_{2^n}'\in\Y^{z_{2^n}}(W)$.
\end{itemize}

\vspace*{3mm}

Now fix $\hat{x}\in G_1$. If $\hat{p}_{y,z}(\hat{x})=0$ for every $(y,z)\in\YZ(W)$, then we clearly have 
$$\prod_{i=1}^{2^n} \hat{p}_{y_i,z_i}(\hat{x})=\prod_{i=1}^{2^n} \hat{p}_{y_i',z_i}(\hat{x}).$$

Therefore, we can assume without loss of generality that there exists $(y,z)\in\YZ(W)$ which satisfies $\hat{p}_{y,z}(\hat{x})\neq0$.

Let $U_0^{2^{n+1}-1}$, $V_0^{2^{n+1}-1}$, $X_0^{2^{n+1}-1}$, $Y_0^{2^{n+1}-1}$ and $Z_0^{2^{n+1}-1}$ be as in Lemma \ref{lemlem} and let $N=2^{n+1}$ so that we have
\begin{equation}
\label{eqCondMinus}
I(U_0;V_1^{N-1}|Z_0^{N-1}V_0)=0.
\end{equation}

Since $X_0^{N-1}=F^{\otimes(n+1)}\cdot U_0^{N-1}$ and $Y_0^{N-1}=F^{\otimes(n+1)}\cdot V_0^{N-1}$, Lemma \ref{lemF} implies that 
\begin{equation}
\label{eqCondMinusI}
U_0=\sum_{i=0}^{N-1}(-1)^{|i|_b}X_i\;\;\text{and}\;\;V_0=\sum_{i=0}^{N-1}(-1)^{|i|_b}Y_i.
\end{equation}

Notice that $\big|\big\{0\leq i< N=2^{n+1}:\; |i|_b\equiv0\bmod 2\big\}\big|=\big|\big\{0\leq i< N=2^{n+1}:\; |i|_b\equiv1\bmod 2\big\}\big|=2^n$. Let $k_1,\ldots,k_{2^n}$ be the elements of $\big\{0\leq i< N:\; |i|_b\equiv0\bmod 2\big\}$ and let $l_1,\ldots,l_{2^n}$ be the elements of $\big\{0\leq i< N:\; |i|_b\equiv1\bmod2\big\}$. 

Define $(\tilde{y}_i,\tilde{y}_i',\tilde{z}_i)_{0\leq i<N}$ as follows:
\begin{itemize}
\item For every $1\leq i\leq 2^n$, let $\tilde{y}_{k_i}=y_i$, $\tilde{y}_{k_i}'=y_i'$ and $\tilde{z}_{k_i}=z_i$.
\item For every $1\leq i\leq 2^n$, let $\tilde{y}_{l_i}=\tilde{y}_{l_i}'=y$ and $\tilde{z}_{l_i}=z$ (where $(y,z)$ is any fixed pair in $\YZ(W)$ satisfying $\hat{p}_{y,z}(\hat{x})\neq0$).
\end{itemize}

Now let $\tilde{v}_0^{N-1} =(F^{\otimes(n+1)})^{-1}\cdot\tilde{y}_0^{N-1}$ and $\tilde{v}'^{N-1}_0 =(F^{\otimes(n+1)})^{-1}\cdot\tilde{y}_0'^{N-1}$. We have
\begin{align*}
\tilde{v}_0&\stackrel{(a)}{=}\sum_{i=0}^{N-1}(-1)^{|i|_b}\tilde{y}_i=\sum_{i=1}^{2^n}(\tilde{y}_{k_i}-\tilde{y}_{l_i})=\left(\sum_{i=1}^{2^n}y_i\right)-2^n y\\
& \stackrel{(b)}{=}\left(\sum_{i=1}^{2^n}y_i'\right)-2^n y=\sum_{i=1}^{2^n}(\tilde{y}_{k_i}'-\tilde{y}_{l_i}')=\sum_{i=0}^{N-1}(-1)^{|i|_b}\tilde{y}_i'\stackrel{(c)}{=}\tilde{v}_0',
\end{align*}
where (a) and (c) follow from Lemma \ref{lemF}. (b) follows from the fact that $\displaystyle\sum_{i=1}^{2^n}y_i=\sum_{i=1}^{2^n}y_i'$. Therefore,
\begin{equation}
\label{eqCondMinusII}
(\tilde{v}_0,\tilde{z}_0^{N-1})=(\tilde{v}_0',\tilde{z}_0^{N-1}).
\end{equation}

On the other hand, since $\tilde{y}_i\in\Y^{\tilde{z}_i}(W)$ for every $0\leq i<N$, we have
\begin{equation}
\label{eqCondMinusIII}
\begin{aligned}
P_{V_0,V_1^{N-1},Z_0^{N-1}}(\tilde{v}_0,\tilde{v}_1^{N-1},\tilde{z}_0^{N-1})&=P_{V_0^{N-1},Z_0^{N-1}}(\tilde{v}_0^{N-1},\tilde{z}_0^{N-1})\\
&=P_{Y_0^{N-1},Z_0^{N-1}}(\tilde{y}_0^{N-1},\tilde{z}_0^{N-1})>0.
\end{aligned}
\end{equation}

Similarly, since $\tilde{y}_i'\in\Y^{\tilde{z}_i}(W)$ for every $0\leq i<N$, we have
\begin{equation}
\label{eqCondMinusIV}
\begin{aligned}
P_{V_0,V_1^{N-1},Z_0^{N-1}}(\tilde{v}_0',\tilde{v}_1'^{N-1},\tilde{z}_0^{N-1})&=P_{V_0^{N-1},Z_0^{N-1}} \big(\tilde{v}'^{N-1}_0,\tilde{z}_0^{N-1}\big)\\
&=P_{Y_0^{N-1},Z_0^{N-1}}(\tilde{y}_0'^{N-1},\tilde{z}_0^{N-1})>0.
\end{aligned}
\end{equation}

Equation \eqref{eqCondMinus} implies that given $(V_0,Z_0^{N-1})$, $U_0$ is conditionally independent of $V_1^{N-1}$. Equations \eqref{eqCondMinusII}, \eqref{eqCondMinusIII} and \eqref{eqCondMinusIV} now imply that for every $u_0\in G_1$, we have:
\begin{align}
&P_{U_0|V_1^{N-1},V_0,Z_0^{N-1}}(u_0|\tilde{v}_1^{N-1},\tilde{v}_0,\tilde{z}_0^{N-1})=P_{U_0|V_1^{N-1},V_0,Z_0^{N-1}}(u_0|\tilde{v}'^{N-1}_1,\tilde{v}_0',\tilde{z}_0^{N-1})\nonumber\\
&\;\;\;\;\;\Leftrightarrow \;\; P_{U_0|V_0^{N-1},Z_0^{N-1}}(u_0|\tilde{v}_0^{N-1},\tilde{z}_0^{N-1})=P_{U_0|V_0^{N-1},Z_0^{N-1}}(u_0|\tilde{v}'^{N-1}_0,\tilde{z}_0^{N-1})\nonumber\\
&\;\;\;\;\;\Leftrightarrow\;\; P_{U_0|Y_0^{N-1},Z_0^{N-1}}(u_0|\tilde{y}_0^{N-1},\tilde{z}_0^{N-1})=P_{U_0|Y_0^{N-1},Z_0^{N-1}}(u_0|\tilde{y}_0'^{N-1},\tilde{z}_0^{N-1})\nonumber\\
&\;\;\;\;\;\stackrel{(a)}{\Leftrightarrow}\;\; \sum_{\substack{\tilde{x}_0^{N-1}\in G_1^{N}:\\\sum_{i=0}^{N-1}(-1)^{|i|_b}\tilde{x}_i=u_0}} \prod_{i=0}^{N-1}P_{X_i|Y_i,Z_1}(\tilde{x}_i|\tilde{y}_i,\tilde{z}_i)=\sum_{\substack{\tilde{x}_0^{N-1}\in G_1^{N}:\\\sum_{i=0}^{N-1}(-1)^{|i|_b}\tilde{x}_i=u_0}} \prod_{i=0}^{N-1}P_{X_i|Y_i,Z_1}(\tilde{x}_i|\tilde{y}_i',\tilde{z}_i)\nonumber\\
&\;\;\;\;\;\stackrel{(b)}{\Leftrightarrow}\;\; \sum_{\substack{x_1^{N}\in G_1^{N}:\\\sum_{i=1}^{2^n}x_i-\sum_{i=2^n+1}^{N}x_i=u_0}} \prod_{i=1}^{2^n}p_{y_i,z_i}(x_i)\prod_{i=2^n+1}^{N}p_{y,z}(x_i)\nonumber\\
&\;\;\;\;\;\;\;\;\;\;\;\;\;\;\;\;\;\;\;\;\;\;\;\;\;\;\;\;\;\;\;\;\;\;\;\;\;\;\;\;\;\;\;\;\;\;\;\;\;\;\;\;\;\;\;\;\;\;\;\;\;\;\;\;\;\;=\sum_{\substack{x_1^{N}\in G_1^{N}:\\\sum_{i=1}^{2^n}x_i-\sum_{i=2^n+1}^{N}x_i=u_0}} \prod_{i=1}^{2^n}p_{y_i',z_i}(x_i)\prod_{i=2^n+1}^{N}p_{y,z}(x_i),
\label{eqCondMinusV}
\end{align}
where (a) follows from \eqref{eqCondMinusI} and (b) follows from the following change of variables:
$$x_i=\begin{cases}\tilde{x}_{k_i}\;&\text{if}\;1\leq i\leq 2^n,\\
\tilde{x}_{l_{i-2^n}}\;&\text{if}\;2^n\leq i\leq 2^{n+1}=N.\end{cases}$$

Now notice that the left hand side of \eqref{eqCondMinusV} is the convolution of $(p_{y_i,z_i})_{1\leq i\leq 2^n}$ with $2^n$ copies of $\tilde{p}_{y,z}$ (where we define $\tilde{p}_{y,z}(x)=p_{y,z}(-x)$). Likewise, the right hand side of \eqref{eqCondMinusV} is the convolution of $(p_{y_i',z_i})_{1\leq i\leq 2^n}$ with $2^n$ copies of $\tilde{p}_{y,z}$. By applying the DFT on \eqref{eqCondMinusV}, we get:
\begin{align*}
\prod_{i=1}^{2^n}\hat{p}_{y_i,z_i}(\hat{u}_0)\prod_{i=2^n+1}^{N}\hat{p}_{y,z}(\hat{u}_0)^{\ast}=\prod_{i=1}^{2^n}\hat{p}_{y_i',z_i}(\hat{u}_0)\prod_{i=2^n+1}^{N}\hat{p}_{y,z}(\hat{u}_0)^{\ast},\;\;\forall\hat{u}_0\in G_1.
\end{align*}

In particular,
\begin{align*}
\prod_{i=1}^{2^n}\hat{p}_{y_i,z_i}(\hat{x})\prod_{i=2^n+1}^{N}\hat{p}_{y,z}(\hat{x})^{\ast}=\prod_{i=1}^{2^n}\hat{p}_{y_i',z_i}(\hat{x})\prod_{i=2^n+1}^{N}\hat{p}_{y,z}(\hat{x})^{\ast}.
\end{align*}

Now since $\hat{p}_{y,z}(\hat{x})\neq0$, we conclude that
$$\prod_{i=1}^{2^n} \hat{p}_{y_i,z_i}(\hat{x})=\prod_{i=1}^{2^n} \hat{p}_{y_i',z_i}(\hat{x}).$$

\section{Proof of Proposition \ref{propWminus}}
\label{apppropWminus}

We need the following few lemmas.

\begin{mylem}
\label{LemmaeqMinusLala0}
For every two-user MAC $W:G_1\times G_2\longrightarrow\mathcal{Z}$ and every $z_1,z_2\in\mathcal{Z}$, we have:
$$\textstyle\Y^{(z_1,z_2)}(W^-)=\Y^{z_1}(W)-\Y^{z_2}(W)=\big\{y_1-y_2:\;y_1\in\textstyle\Y^{z_1}(W),y_2\in\textstyle\Y^{z_2}(W)\big\}.$$
\end{mylem}
\begin{proof}
Let $U_1,U_2,V_1,V_2,X_1,X_2,Y_1,Y_2,Z_1,Z_2$ be as in Remark \ref{MainRem}. For every $v_1\in G_2$ and every $z_1,z_2\in\mathcal{Z}$, we have:
$$P_{V_1,Z_1,Z_2}(v_1,z_1,z_2)=\sum_{\substack{y_1,y_2\in G_2:\\v_1=y_1-y_2}}P_{Y_1,Y_2,Z_1,Z_2}(y_1,y_2,z_1,z_2)=\sum_{\substack{y_1,y_2\in G_2:\\v_1=y_1-y_2}}P_{Y_1,Z_1}(y_1,z_1)P_{Y_2,Z_2}(y_2,z_2).$$
Therefore, $v_1\in\Y^{(z_1,z_2)}(W^-)$ if and only if there exist $y_1,y_2\in G_2$ such that $y_1\in\Y^{z_1}(W)$, $y_2\in\Y^{z_2}(W)$ and $v_1=y_1-y_2$. Hence,
\begin{equation*}
\textstyle\Y^{(z_1,z_2)}(W^-)=\big\{y_1-y_2:\;y_1\in\textstyle\Y^{z_1}(W), y_2\in\textstyle\Y^{z_2}(W)\big\}.
\end{equation*}
\end{proof}

\begin{mylem}
Let $U_1,U_2,V_1,V_2,X_1,X_2,Y_1,Y_2,Z_1,Z_2$ be as in Remark \ref{MainRem}. For every $z_1,z_2\in\mathcal{Z}$, every $v_1\in\Y^{(z_1,z_2)}(W^-)$ and every $\hat{u}_1\in G_1$, we have:
\begin{equation}
\hat{p}_{v_1,(z_1,z_2),W^-}(\hat{u}_1)=\sum_{\substack{v_2\in \Y^{z_2}(W):\\v_1+v_2\in\Y^{z_1}(W)}}\frac{P_{Y_1|Z_1}(v_1+v_2|z_1)P_{Y_2|Z_2}(v_2|z_2)}{P_{V_1|Z_1,Z_2}(v_1|z_1,z_2)}\hat{p}_{v_1+v_2,z_1}(\hat{u}_1)\cdot \hat{p}_{v_2,z_2}(\hat{u}_1)^\ast.
\label{eqMinusLala}
\end{equation}
\end{mylem}
\begin{proof}
Fix $z_1,z_2\in\mathcal{Z}$ and $v_1\in\Y^{(z_1,z_2)}(W^-)$, and let $\beta=P_{V_1|Z_1,Z_2}(v_1|z_1,z_2)>0$. For every $u_1\in G_1$, we have:
\begin{align*}
p_{v_1,(z_1,z_2),W^-}(u_1)&=P_{U_1|V_1,Z_1,Z_2}(u_1|v_1,z_1,z_2)=\frac{1}{\beta}P_{U_1,V_1|Z_1,Z_2}(u_1,v_1|z_1,z_2)\\
&=\frac{1}{\beta}\sum_{\substack{u_2\in G_1,\\v_2\in G_2}} P_{U_1,U_2,V_1,V_2|Z_1,Z_2}(u_1,u_2,v_1,v_2|z_1,z_2)\\
&=\frac{1}{\beta}\sum_{\substack{u_2\in G_1,\\v_2\in G_2}} P_{X_1,X_2,Y_1,Y_2|Z_1,Z_2}(u_1+u_2,u_2,v_1+v_2,v_2|z_1,z_2)\\
&=\frac{1}{\beta}\sum_{\substack{v_2\in G_2}}\sum_{\substack{u_2\in G_1}} P_{X_1,Y_1|Z_1}(u_1+u_2,v_1+v_2|z_1)P_{X_2,Y_2|Z_2}(u_2,v_2|z_2)\\
&=\frac{1}{\beta}\sum_{\substack{v_2\in \Y^{z_2}(W):\\v_1+v_2\in\Y^{z_1}(W)}}\sum_{u_2\in G_1} P_{X_1,Y_1|Z_1}(u_1+u_2,v_1+v_2|z_1)P_{X_2,Y_2|Z_2}(u_2,v_2|z_2)\\
&=\frac{1}{\beta}\sum_{\substack{v_2\in \Y^{z_2}(W):\\v_1+v_2\in\Y^{z_1}(W)}}P_{Y_1|Z_1}(v_1+v_2|z_1)P_{Y_2|Z_2}(v_2|z_2)\sum_{u_2\in G_1}p_{v_1+v_2,z_1}(u_1+u_2)p_{v_2,z_2}(u_2)\\
&=\frac{1}{\beta}\sum_{\substack{v_2\in \Y^{z_2}(W):\\v_1+v_2\in\Y^{z_1}(W)}}P_{Y_1|Z_1}(v_1+v_2|z_1)P_{Y_2|Z_2}(v_2|z_2)(p_{v_1+v_2,z_1}\ast \tilde{p}_{v_2,z_2})(u_1),
\end{align*}
where we define $\tilde{p}_{v_2,z_2}(x)=p_{v_2,z_2}(-x)$ for every $x\in G_1$. Therefore, for every $\hat{u}_1\in G_1$, we have:
\begin{equation*}
\hat{p}_{v_1,(z_1,z_2),W^-}(\hat{u}_1)=\sum_{\substack{v_2\in \Y^{z_2}(W):\\v_1+v_2\in\Y^{z_1}(W)}}\frac{P_{Y_1|Z_1}(v_1+v_2|z_1)P_{Y_2|Z_2}(v_2|z_2)}{P_{V_1|Z_1,Z_2}(v_1|z_1,z_2)}\hat{p}_{v_1+v_2,z_1}(\hat{u}_1)\cdot \hat{p}_{v_2,z_2}(\hat{u}_1)^\ast.
\end{equation*}
\end{proof}

\begin{mylem}
If $I_1$ is $\ast^-$ preserved for $W$, then
$\XDY(W^-)\subset \{(\hat{x},y_1+y_2):\;(\hat{x},y_1),(\hat{x},y_2)\in\XDY(W)\}$.
\label{LemmaMinusLala}
\end{mylem}
\begin{proof}
Let $U_1,U_2,V_1,V_2,X_1,X_2,Y_1,Y_2,Z_1,Z_2$ be as in Remark \ref{MainRem}. Let $(\hat{u}_1,v_1)\in\XDY(W^-)$. There exists $z^-=(z_1,z_2)\in\mathcal{Z}^2$ such that $(\hat{u}_1,v_1)\in\XDY^{z^-}(W^-)$, i.e., $\hat{u}_1\in\X^{z^-}(W^-)$ and $v_1\in\DY^{z^-}(W^-)$. This implies the existence of $v_1',v_1''\in\Y^{z^-}(W^-)$ such that $v_1=v_1'-v_1''$. Since $\hat{u}_1\in\X^{z^-}(W^-)$, Lemma \ref{lem3} shows that $\hat{p}_{v_1',z^-,W^-}(\hat{u}_1)\neq 0$ and $\hat{p}_{v_1'',z^-,W^-}(\hat{u}_1)\neq 0$. From \eqref{eqMinusLala}, we have:
\begin{equation*}
\hat{p}_{v_1',z^-,W^-}(\hat{u}_1)=\sum_{\substack{v_2'\in \Y^{z_2}(W):\\v_1'+v_2'\in\Y^{z_1}(W)}}\frac{P_{Y_1|Z_1}(v_1'+v_2'|z_1)P_{Y_2|Z_2}(v_2'|z_2)}{P_{V_1|Z_1,Z_2}(v_1'|z_1,z_2)}\hat{p}_{v_1'+v_2',z_1}(\hat{u}_1)\cdot \hat{p}_{v_2',z_2}(\hat{u}_1)^\ast.
\end{equation*}
Since $\hat{p}_{v_1',z^-,W^-}(\hat{u}_1)\neq 0$, the terms in the above sum cannot all be zero. Therefore, there exists $v_2'\in\Y^{z_2}(W)$ such that $v_1'+v_2'\in\Y^{z_1}(W)$, $\hat{p}_{v_1'+v_2',z_1}(\hat{u}_1)\neq 0$ and $\hat{p}_{v_2',z_2}(\hat{u}_1)\neq 0$. Similarly, since $\hat{p}_{v_1'',z^-,W^-}(\hat{u}_1)\neq 0$, there exists $v_2''\in\Y^{z_2}(W)$ such that $v_1''+v_2''\in\Y^{z_1}(W)$, $\hat{p}_{v_1''+v_2'',z_1}(\hat{u}_1)\neq 0$ and $\hat{p}_{v_2'',z_2}(\hat{u}_1)\neq 0$. Therefore, we have
\begin{itemize}
\item $\hat{u}_1\in\X^{z_1}(W)$ (because $\hat{p}_{v_1'+v_2',z_1}(\hat{u}_1)\neq 0$).
\item $v_1+v_2'-v_2''=(v_1'+v_2')-(v_1''+v_2'')\in\DY^{z_1}(W)$.
\item $\hat{u}_1\in\X^{z_2}(W)$ (because $\hat{p}_{v_2',z_2}(\hat{u}_1)\neq 0$).
\item $v_2''-v_2'\in\DY^{z_2}(W)$.
\end{itemize}

We can now see that $(\hat{u}_1,v_1+v_2'-v_2'')\in \XDY^{z_1}(W)\subset \XDY(W)$ and $(\hat{u}_1, v_2''-v_2')\in \XDY^{z_2}(W)\subset \XDY(W)$. By noticing that $v_1=(v_1+v_2'-v_2'')+(v_2''-v_2')$, we conclude that:
$$\XDY(W^-)\subset\big\{(\hat{x},y_1+y_2):\;(\hat{x},y_1),(\hat{x},y_2)\in\XDY(W)\big\}.$$
\end{proof}

Now we are ready to prove Proposition \ref{propWminus}.

Let $W$ be a two-user MAC such that $I_1$ is $\ast^-$ preserved for $W$. Let $U_1,U_2,V_1,V_2,X_1,X_2,Y_1,Y_2,Z_1,Z_2$ be as in Remark \ref{MainRem}.
\begin{enumerate}

\item Let $\hat{x}\in G_1$ and $y_1,y_2\in G_2$ be such that $(\hat{x},y_1),(\hat{x},y_2)\in\XDY(W)$. There exist $z_1,z_2\in\mathcal{Z}$, $y_1',y_1''\in\Y^{z_1}(W)$ and $y_2',y_2''\in\Y^{z_2}(W)$ such that $\hat{x}\in\X^{z_1}(W)$, $\hat{x}\in\X^{z_2}(W)$, $y_1=y_1'-y_1''$ and $y_2=y_2'-y_2''$. Lemma \ref{lem3} implies that $\hat{p}_{y_1',z_1}(\hat{x})\neq0$, $\hat{p}_{y_1'',z_1}(\hat{x})\neq0$, $\hat{p}_{y_2',z_2}(\hat{x})\neq0$ and $\hat{p}_{y_2'',z_2}(\hat{x})\neq0$. Now from Lemma \ref{LemmaeqMinusLala0} we get $y_1'-y_2''\in \Y^{(z_1,z_2)}(W^-)$ and $y_1''-y_2'\in \Y^{(z_1,z_2)}(W^-)$.

For every $v_2\in \Y^{z_2}(W)$ satisfying $y_1'-y_2''+v_2\in\Y^{z_1}(W)$, we have:
\begin{equation}
\label{eqHsfsd124f78y34}
\begin{aligned}
\hat{p}_{y_1'-y_2''+v_2,z_1}(\hat{x})\cdot \hat{p}_{v_2,z_2}(\hat{x})^\ast&=\hat{p}_{y_1',z_1}(\hat{x})\hat{f}_W(\hat{x},v_2-y_2'')\cdot\hat{p}_{y_2'',z_2}(\hat{x})^\ast\hat{f}_W(\hat{x},v_2-y_2'')^{\ast}\\
&\stackrel{(a)}{=}\hat{p}_{y_1',z_1}(\hat{x})\hat{p}_{y_2'',z_2}(\hat{x})^\ast,
\end{aligned}
\end{equation}
where (a) follows from the fact that $\hat{f}_W(\hat{x},v_2-y_2'')\in\mathbb{T}$, which means that $$\hat{f}_W(\hat{x},v_2-y_2'')\hat{f}_W(\hat{x},v_2-y_2'')^{\ast}=|\hat{f}_W(\hat{x},v_2-y_2'')|^2=1.$$

Let $z^-=(z_1,z_2)\in\mathcal{Z}^2$. From \eqref{eqMinusLala}, we have:
\begin{align*}
&\hat{p}_{y_1'-y_2'',z^-,W^-}(\hat{x})\\
&\;\;\;\;\;\;\;=\sum_{\substack{v_2\in \Y^{z_2}(W):\\y_1'-y_2''+v_2\in\Y^{z_1}(W)}}\frac{P_{Y_1|Z_1}(y_1'-y_2''+v_2|z_1)P_{Y_2|Z_2}(v_2|z_2)}{P_{V_1|Z_1,Z_2}(y_1'-y_2''|z_1,z_2)}\hat{p}_{y_1'-y_2''+v_2,z_1}(\hat{x})\cdot \hat{p}_{v_2,z_2}(\hat{x})^\ast\\
&\;\;\;\;\;\;\;\stackrel{(a)}{=}\sum_{\substack{v_2\in \Y^{z_2}(W):\\y_1'-y_2''+v_2\in\Y^{z_1}(W)}}\frac{P_{Y_1|Z_1}(y_1'-y_2''+v_2|z_1)P_{Y_2|Z_2}(v_2|z_2)}{P_{V_1|Z_1,Z_2}(y_1'-y_2''|z_1,z_2)}\hat{p}_{y_1',z_1}(\hat{x})\hat{p}_{y_2'',z_2}(\hat{x})^\ast\\
&\;\;\;\;\;\;\;=\hat{p}_{y_1',z_1}(\hat{x})\hat{p}_{y_2'',z_2}(\hat{x})^\ast\sum_{\substack{v_2\in \Y^{z_2}(W):\\y_1'-y_2''+v_2\in\Y^{z_1}(W)}}\frac{P_{Y_1|Z_1}(y_1'-y_2''+v_2|z_1)P_{Y_2|Z_2}(v_2|z_2)}{P_{V_1|Z_1,Z_2}(y_1'-y_2''|z_1,z_2)}\\
&\;\;\;\;\;\;\;=\hat{p}_{y_1',z_1}(\hat{x})\hat{p}_{y_2'',z_2}(\hat{x})^\ast\neq 0,
\end{align*}
where (a) follows from \eqref{eqHsfsd124f78y34}. This shows that $\hat{x}\in\X^{z^-}(W^-)$. Now since $y_1'-y_2''\in\Y^{z^-}(W^-)$ and $y_1''-y_2'\in\Y^{z^-}(W^-)$, we have $(y_1'-y_2'')-(y_1''-y_2')\in\DY^{z^-}(W^-)$. Therefore,
$$(\hat{x},y_1+y_2)=(\hat{x},y_1'-y_1''+y_2'-y_2'')=\big(\hat{x},(y_1'-y_2'')-(y_1''-y_2')\big)\in\XDY(W^-).$$
Hence, $\big\{(\hat{x},y_1+y_2):\;(\hat{x},y_1),(\hat{x},y_2)\in\XDY(W)\big\}\subset \XDY(W^-)$. We conclude that $$\XDY(W^-)=\big\{(\hat{x},y_1+y_2):\;(\hat{x},y_1),(\hat{x},y_2)\in\XDY(W)\big\}$$ since the other inclusion was shown in Lemma \ref{LemmaMinusLala}.
\item Let $\hat{x},y_1,y_2$ be such that $(\hat{x},y_1),(\hat{x},y_2)\in\XDY(W)$. Define $y_1',y_1'',y_2',y_2'',z_1,z_2,z^-$ as in 1). We have shown that $\hat{p}_{y_1'-y_2'',z^-,W^-}(\hat{x})=\hat{p}_{y_1',z_1}(\hat{x})\hat{p}_{y_2'',z_2}(\hat{x})^\ast$. Similarly, we can show that $\hat{p}_{y_1''-y_2',z^-,W^-}(\hat{x})=\hat{p}_{y_1'',z_1}(\hat{x})\hat{p}_{y_2',z_2}(\hat{x})^\ast$. Therefore,
\begin{align*}
\hat{f}_{W^-}(\hat{x},y_1+y_2)&=\hat{f}_{W^-}(\hat{x},y_1'-y_1''+y_2'-y_2'')=\hat{f}_{W^-}\big(\hat{x},(y_1'-y_2'')-(y_1''-y_2')\big)\\
&=\frac{\hat{p}_{y_1'-y_2'',z^-,W^-}(\hat{x})}{\hat{p}_{y_1''-y_2',z^-,W^-}(\hat{x})}=\frac{\hat{p}_{y_1',z_1}(\hat{x})\hat{p}_{y_2'',z_2}(\hat{x})^\ast}{\hat{p}_{y_1'',z_1}(\hat{x})\hat{p}_{y_2',z_2}(\hat{x})^\ast}=\frac{\hat{f}_W(\hat{x},y_1)}{\hat{f}_W(\hat{x},y_2)^\ast} \\
&\stackrel{(a)}{=}\hat{f}_W(\hat{x},y_1)\cdot \hat{f}_W(\hat{x},y_2),
\end{align*}
where (a) follows from the fact that $\hat{f}_W(\hat{x},y_2)\cdot \hat{f}_W(\hat{x},y_2)^\ast=|\hat{f}_W(\hat{x},y_2)|^2=1$.
\end{enumerate}

\section{Proof of Proposition \ref{propWplus}}

\label{apppropWplus}

We need the following few lemmas.

\begin{mylem}
\label{lemPlushaha}
For every $y_1,y_2\in G_2$ and every $z_1,z_2\in \mathcal{Z}$, we have:
\begin{itemize}
\item If $(y_1,z_1)\notin\YZ(W)$ or $(y_2,z_2)\notin\YZ(W)$, then $\big(y_2,(z_1,z_2,u_1,y_1-y_2)\big)\notin\YZ(W^+)$ for every $u_1\in G_1$.
\item If $(y_1,z_1)\in\YZ(W)$ and $(y_2,z_2)\in\YZ(W)$, there exists $u_1\in G_1$ such that $\big(y_2,(z_1,z_2,u_1,y_1-y_2)\big)\in\YZ(W^+)$.
\end{itemize}
\end{mylem}
\begin{proof}
Let $U_1,U_2,V_1,V_2,X_1,X_2,Y_1,Y_2,Z_1,Z_2$ be as in Remark \ref{MainRem}. For every $u_1\in G_1$, every $y_1,y_2\in G_2$ and every $z_1,z_2\in\mathcal{Z}$, we have:
\begin{align*}
P_{V_2,Z_1,Z_2,U_1,V_1}(y_2,z_1,z_2,u_1,y_1-y_2)&=\sum_{u_2\in G_1}P_{U_2,V_2,Z_1,Z_2,U_1,V_1}(u_2,y_2,z_1,z_2,u_1,y_1-y_2)\\
&=\sum_{u_2\in G_1}P_{X_1,X_2,Y_1,Y_2,Z_1,Z_2}(u_1+u_2,u_2,y_1,y_2,z_1,z_2)\\
&=\sum_{u_2\in G_1}P_{X_1,Y_1,Z_1}(u_1+u_2,y_1,z_1)\cdot P_{X_2,Y_2,Z_2}(u_2,y_2,z_2).
\end{align*}
Therefore, we have:
\begin{itemize}
\item If $(y_1,z_1)\notin\YZ(W)$ or $(y_2,z_2)\notin\YZ(W)$, then for all $u_1,u_2\in G_1$, we have $P_{X_1,Y_1,Z_1}(u_1+u_2,y_1,z_1)\leq P_{Y_1,Z_1}(y_1,z_1)=0$ or  $P_{X_2,Y_2,Z_2}(u_2,y_2,z_2)\leq P_{Y_2,Z_2}(y_2,z_2)=0$, which means that $P_{V_2,Z_1,Z_2,U_1,V_1}(y_2,z_1,z_2,u_1,y_1-y_2)=0$. Hence $\big(y_2,(z_1,z_2,u_1,y_1-y_2)\big)\notin\YZ(W^+)$ for every $u_1\in G_1$.
\item If $(y_1,z_1)\in\YZ(W)$ and $(y_2,z_2)\in\YZ(W)$, then $P_{Y_1,Z_1}(y_1,z_1)>0$ and $P_{Y_2,Z_2}(y_2,z_2)>0$. This means that there exist $x_1,x_2\in G_1$ such that $P_{X_1,Y_1,Z_1}(x_1,y_1,z_1)>0$ and $P_{X_2,Y_2,Z_2}(x_2,y_2,z_2)>0$. Let $u_1=x_1-x_2$ and $u_2=x_2$. We have $P_{X_1,Y_1,Z_1}(u_1+u_2,y_1,z_1)\cdot P_{X_2,Y_2,Z_2}(u_2,y_2,z_2)>0$, which implies that $P_{V_2,Z_1,Z_2,U_1,V_1}(y_2,z_1,z_2,u_1,y_1-y_2)>0$ hence $\big(y_2,(z_1,z_2,u_1,y_1-y_2)\big)\in\YZ(W^+)$.
\end{itemize}
\end{proof}

\begin{mylem}
Let $U_1,U_2,V_1,V_2,X_1,X_2,Y_1,Y_2,Z_1,Z_2$ be as in Remark \ref{MainRem}. For every $\big(v_2,(z_1,z_2,u_1,v_1)\big)\in\YZ(W^+)$, we have:
\vspace*{-1mm}
\begin{equation}
\hat{p}_{v_2,(z_1,z_2,u_1,v_1),W^+}(\hat{u}_2)=\sum_{\hat{u}_2'\in G_1}\frac{\hat{p}_{v_1+v_2,z_1}(\hat{u}_2')\cdot\hat{p}_{v_2,z_2}(\hat{u}_2-\hat{u}_2')}{|G_1|\alpha(u_1,z_1,z_2,v_1,v_2)}e^{j2\pi\langle \hat{u}_2',u_1 \rangle},
\label{eqleqTek}
\end{equation}
where $\alpha(u_1,z_1,z_2,v_1,v_2)=P_{U_1|Z_1,Z_2,V_1,V_2}(u_1|z_1,z_2,v_1,v_2)$.
\end{mylem}
\begin{proof}
For every $\big(v_2,(z_1,z_2,u_1,v_1)\big)\in\YZ(W^+)$ and every $u_2\in G_2$, we have:
\begin{align*}
p_{v_2,(z_1,z_2,u_1,v_1),W^+}(u_2)&=P_{U_2|V_2,Z_1,Z_2,U_1,V_1}(u_2|v_2,z_1,z_2,u_1,v_1)\\
&=\frac{P_{U_1,U_2|Z_1,Z_2,V_1,V_2}(u_1,u_2|z_1,z_2,v_1,v_2)}{P_{U_1|Z_1,Z_2,V_1,V_2}(u_1|z_1,z_2,v_1,v_2)}\\
&=\frac{P_{X_1,X_2|Z_1,Z_2,Y_1,Y_2}(u_1+u_2,u_2|z_1,z_2,v_1+v_2,v_2)}{\alpha(u_1,z_1,z_2,v_1,v_2)}\\
&=\frac{P_{X_1|Z_1,Y_1}(u_1+u_2|z_1,v_1+v_2)P_{X_2|Z_2,Y_2}(u_2|z_2,v_2)}{\alpha(u_1,z_1,z_2,v_1,v_2)}\\
&=\frac{p_{v_1+v_2,z_1}(u_1+u_2)p_{v_2,z_2}(u_2)}{\alpha(u_1,z_1,z_2,v_1,v_2)}.
\end{align*}
Therefore, for every $\hat{u}_2\in G_2$, we have:
\vspace*{-1mm}
\begin{align*}
\hat{p}_{v_2,(z_1,z_2,u_1,v_1),W^+}(\hat{u}_2)&=\frac{\frac{1}{|G_1|}\left( \hat{p}_{v_1+v_2,z_1}(\hat{u}_2)e^{j2\pi\langle \hat{u}_2,u_1 \rangle}\right)\ast \hat{p}_{v_2,z_2}(\hat{u}_2)}{\alpha(u_1,z_1,z_2,v_1,v_2)}\\
&=\frac{\sum_{\hat{u}_2'\in G_1}\hat{p}_{v_1+v_2,z_1}(\hat{u}_2')e^{j2\pi\langle \hat{u}_2',u_1 \rangle}\hat{p}_{v_2,z_2}(\hat{u}_2-\hat{u}_2')}{|G_1|\alpha(u_1,z_1,z_2,v_1,v_2)}\\
&=\sum_{\hat{u}_2'\in G_1}\frac{\hat{p}_{v_1+v_2,z_1}(\hat{u}_2')\cdot\hat{p}_{v_2,z_2}(\hat{u}_2-\hat{u}_2')}{|G_1|\alpha(u_1,z_1,z_2,v_1,v_2)}e^{j2\pi\langle \hat{u}_2',u_1 \rangle}.
\end{align*}
\end{proof}

\begin{mylem}
\label{lemTata}
Let $(y_1,z_1),(y_2,z_2)\in\YZ(W)$ and $\hat{x}\in G_1$. If there exists $u_1\in G_1$ such that
\begin{equation}
\label{eqleqTek1}
\sum_{\hat{u}\in G_1}\hat{p}_{y_1,z_1}(\hat{u})\cdot\hat{p}_{y_2,z_2}(\hat{x}-\hat{u})e^{j2\pi\langle \hat{u},u_1\rangle}\neq 0,
\end{equation}
then we have:
\begin{itemize}
\item $(y_2,z^+)\in\YZ(W^+)$, where $z^+=(z_1,z_2,u_1,y_1-y_2)$.
\item $\hat{x}\in \X^{z^+}(W^+)$.
\end{itemize}
\end{mylem}
\begin{proof}
Let $U_1,U_2,V_1,V_2,X_1,X_2,Y_1,Y_2,Z_1,Z_2$ be as in Remark \ref{MainRem}. Let $v_1=y_1-y_2$ and $v_2=y_2$. Notice that the expression in \eqref{eqleqTek1} is the DFT of the mapping $K:G_1\rightarrow\mathbb{C}$ defined as $$K(x)=|G_1|\cdot p_{y_1,z_1}(u_1+x)\cdot p_{y_2,z_2}(x).$$ Equation \eqref{eqleqTek1} shows that $\hat{K}$ is not zero everywhere which implies that $K$ is not zero everywhere. Therefore, there exists $x\in G_1$ such that $K(x)\neq 0$. We have:
\begin{align*}
P_{V_2,Z_1,Z_2,U_1,V_1}(v_2,z_1,z_2,u_1,v_1)&\geq P_{U_1,U_2,V_1,V_2,Z_1,Z_2}(u_1,x,y_1-y_2,y_2,z_1,z_2)\\
&=P_{X_1,X_2,Y_1,Y_2,Z_1,Z_2}(u_1+x,x,y_1,y_2,z_1,z_2)\\
&=P_{X_1,Y_1,Z_1}(u_1+x,y_1,z_1)P_{X_2,Y_2,Z_2}(x,y_2,z_2)\\
&=P_{Y_1,Z_1}(y_1,z_1)p_{y_1,z_1}(u_1+x)\cdot P_{Y_2,Z_2}(y_2,z_2)p_{y_2,z_2}(x)\\
&=P_{Y_1,Z_1}(y_1,z_1)\cdot P_{Y_2,Z_2}(y_2,z_2)\cdot \frac{K(x)}{|G_1|}\stackrel{(a)}{>}0,
\end{align*}
where (a) follows from the fact that $y_1\in\Y^{z_1}(W)$, $y_2\in\Y^{z_2}(W)$ and $K(x)>0$. We conclude that $\big(v_2,(z_1,z_2,u_1,v_1)\big)\in\YZ(W^+)$ and so we can apply \eqref{eqleqTek} to $(v_2,z_1,z_2,u_1,v_1)$:
\begin{align*}
\hat{p}_{v_2,(z_1,z_2,u_1,v_1),W^+}(\hat{x})&\stackrel{(a)}{=}\sum_{\hat{u}\in G_1}\frac{\hat{p}_{y_1,z_1}(\hat{u})\cdot\hat{p}_{y_2,z_2}(\hat{x}-\hat{u})}{|G_1|\alpha(u_1,z_1,z_2,v_1,v_2)}e^{j2\pi\langle \hat{u},u_1 \rangle}\stackrel{(b)}{\neq}0,
\end{align*}
where (b) follows from \eqref{eqleqTek1}. Therefore, $\hat{p}_{y_2,z^+,W^+}(\hat{x})\neq 0$, where $z^+=(z_1,z_2,u_1,y_1-y_2)$. Hence $\hat{x}\in \X^{z^+}(W^+)$.
\end{proof}

\vspace*{3mm}

Now we are ready to prove Proposition \ref{propWplus}.

Let $W$ be a two-user MAC and assume that polarization $\ast$-preserves $I_1$ for $W$. Let $U_1,U_2,V_1,V_2,X_1,X_2,Y_1,Y_2,Z_1,Z_2$ be as in Remark \ref{MainRem}.
\begin{enumerate}
\item Suppose that $\hat{x}_1,\hat{x}_2\in G_1$ and $y\in G_2$ satisfy $(\hat{x}_1,y),(\hat{x}_2,y)\in\XDY(W)$ and let $\hat{x}=\hat{x}_1+\hat{x}_2$. There exist $z_1,z_2\in\mathcal{Z}$, $y_1,y_1'\in\Y^{z_1}(W)$ and $y_2,y_2'\in\Y^{z_2}(W)$ such that
\begin{itemize}
\item $\hat{x}_1\in\X^{z_1}(W)$ and $y=y_1-y_1'$.
\item $\hat{x}_2\in\X^{z_2}(W)$ and $y=y_2-y_2'$.
\end{itemize}
Lemma \ref{lem3} implies that $\hat{p}_{y_1,z_1}(\hat{x}_1)\neq0$, $\hat{p}_{y_1',z_1}(\hat{x}_1)\neq0$, $\hat{p}_{y_2,z_2}(\hat{x}_2)\neq0$ and $\hat{p}_{y_2',z_2}(\hat{x}_2)\neq0$.

Let $v_1=y_1-y_2=y_1'-y_2'$, $v_2=y_2$ and $v_2'=y_2'$. Define the mapping $\hat{L}:G_1\rightarrow\mathbb{C}$ as $$\hat{L}(\hat{u})=\hat{p}_{y_1,z_1}(\hat{u})\cdot\hat{p}_{y_2,z_2}(\hat{x}-\hat{u}).$$ We have:
$\hat{L}(\hat{x}_1)=\hat{p}_{y_1,z_1}(\hat{x}_1)\cdot\hat{p}_{y_2,z_2}(\hat{x}_2)\neq0$. Therefore, the mapping $\hat{L}$ is not zero everywhere, which implies that its inverse DFT is not zero everywhere. Hence there exists $u_1\in G_1$ such that:
\begin{equation*}
\sum_{\hat{u}\in G_1}\hat{p}_{y_1,z_1}(\hat{u})\cdot\hat{p}_{y_2,z_2}(\hat{x}-\hat{u})e^{j2\pi\langle \hat{u},u_1\rangle}\neq 0.
\end{equation*}

It follows from Lemma \ref{lemTata} that $(v_2,z^+)\in\YZ(W^+)$ and  $\hat{x}\in\X^{z^+}(W^+)$, where $z^+=(z_1,z_2,u_1,v_1)$. If we can also show that $(v_2',z^+)\in\YZ(W^+)$ we will be able to conclude that $(\hat{x},y)\in\XDY(W^+)$ since $y=v_2-v_2'$. We have the following:
\begin{itemize}
\item $P_{U_1,Z_1,Z_2,V_1}(u_1,z_1,z_2,v_1)\geq P_{V_2,Z_1,Z_2,U_1,V_1}(v_2,z_1,z_2,u_1,v_1)>0$ since $(v_2,z^+)\in\YZ(W^+)$. Hence, $$P_{U_1|Z_1,Z_2,V_1}(u_1|z_1,z_2,v_1)>0.$$
\item $P_{V_2,Z_1,Z_2,V_1}(v_2',z_1,z_2,v_1)=P_{Y_1,Z_1,Y_2,Z_2}(y_1',z_1,y_2',z_2)
>0$ since $y_1'\in\Y^{z_1}(W)$ and $y_2'\in\Y^{z_2}(W)$. Thus, $$P_{V_2|Z_1,Z_2,V_1}(v_2'|z_1,z_2,v_1)>0.$$
\end{itemize}
But $I_1$ is preserved for $W$, so we must have $I(U_1;V_2|Z_1Z_2V_1)=0$. Therefore,
\begin{equation}
P_{U_1,V_2|Z_1,Z_2,V_1}(u_1,v_2'|z_1,z_2,v_1)=P_{U_1|Z_1,Z_2,V_1}(u_1|z_1,z_2,v_1)\cdot P_{V_2|Z_1,Z_2,V_1}(v_2'|z_1,z_2,v_1)>0.
\label{Reason}
\end{equation}
We conclude that $P_{V_2,Z_1,Z_2,U_1,V_1}(v_2',z_1,z_2,u_1,v_1)>0$ and so $(v_2',z^+)\in\YZ(W^+)$. Hence, $(\hat{x},y)\in\XDY(W^+)$. We conclude that $(\hat{x}_1+\hat{x}_2,y)\in\XDY(W^+)$ for every $\hat{x}_1,\hat{x}_2\in G_1$ and every $y\in G_2$ satisfying $(\hat{x}_1,y),(\hat{x}_2,y)\in\XDY(W)$. Therefore, $$\big\{(\hat{x}_1+\hat{x}_2,y):\;(\hat{x}_1,y),(\hat{x}_2,y)\in\XDY(W)\big\}\subset \XDY(W^+).$$
\item Suppose that $\hat{x}_1,\hat{x}_2\in G_1$ and $y\in G_2$ satisfy $(\hat{x}_1,y),(\hat{x}_2,y)\in\XDY(W)$ and let $\hat{x}=\hat{x}_1+\hat{x}_2$. Let $y_1,y_2,y_1',y_2',v_1,v_2,v_2',z_1,z_2,z^+$ be defined as in 1) so that $v_2,v_2'\in\Y^{z+}(W^+)$, $y=v_2-v_2'$ and $\hat{x}\in\X^{z^+}(W^+)$. Lemma \ref{lem3} implies that $\hat{p}_{v_2,z^+,W^+}(\hat{x})\neq 0$ and $\hat{p}_{v_2',z^+,W^+}(\hat{x})\neq 0$. Now since $(\hat{x},y)=(\hat{x},v_2-v_2')\in\XDY(W^+)$, we have:
\begin{equation*}
\begin{aligned}
\hat{p}_{v_2,(z_1,z_2,u_1,v_1),W^+}(\hat{x})&=\hat{p}_{v_2,z^+,W^+}(\hat{x})=\hat{f}_{W^+}(\hat{x},y)\cdot\hat{p}_{v_2',z^+,W^+}(\hat{x})\\
&=\hat{f}_{W^+}(\hat{x},y)\cdot\hat{p}_{v_2',(z_1,z_2,u_1,v_1),W^+}(\hat{x}).
\end{aligned}
\end{equation*}

Define $F:G_1\rightarrow\mathbb{C}$ and $F':G_1\rightarrow\mathbb{C}$ as follows:
\begin{align*}
F(u_1')=\sum_{\hat{u}\in G_1}\hat{p}_{y_1,z_1}(\hat{u})\cdot\hat{p}_{y_2,z_2}(\hat{x}-\hat{u})e^{j2\pi\langle \hat{u},u_1'\rangle}.
\end{align*}
\begin{align*}
F'(u_1')=\sum_{\hat{u}\in G_1}\hat{p}_{y_1',z_1}(\hat{u})\cdot\hat{p}_{y_2',z_2}(\hat{x}-\hat{u})e^{j2\pi\langle \hat{u},u_1'\rangle}.
\end{align*}
For every $u_1'\in G_1$, we have:
\begin{itemize}
\item If $F(u_1')\neq 0$ then $\big(v_2,(z_1,z_2,u_1',v_1)\big)\in\YZ(W^+)$ and $\hat{x}\in\X^{(z_1,z_2,u_1',v_1)}(W^+)$ by Lemma \ref{lemTata}. By replacing $u_1$ by $u_1'$ in \eqref{Reason}, we can get $\big(v_2',(z_1,z_2,u_1',v_1)\big)\in\YZ(W^+)$. Therefore, 
\begin{equation}
\label{sdnhbfvhjwsdfvhs7}
\hat{p}_{v_2,(z_1,z_2,u_1',v_1),W^+}(\hat{x})=\hat{f}_{W^+}(\hat{x},y)\cdot\hat{p}_{v_2',(z_1,z_2,u_1',v_1),W^+}(\hat{x}).
\end{equation}
We have:
\begin{align*}
F(u_1')&=\sum_{\hat{u}\in G_1}\hat{p}_{y_1,z_1}(\hat{u})\cdot\hat{p}_{y_2,z_2}(\hat{x}-\hat{u})e^{j2\pi\langle \hat{u},u_1'\rangle}\\
&\stackrel{(a)}{=}|G_1|\cdot\alpha(u_1',z_1,z_2,v_1,v_2) \hat{p}_{v_2,(z_1,z_2,u_1',v_1),W^+}(\hat{x})\\
&\stackrel{(b)}{=}\frac{\alpha(u_1',z_1,z_2,v_1,v_2)}{\alpha(u_1',z_1,z_2,v_1,v_2')}|G_1|\alpha(u_1',z_1,z_2,v_1,v_2')\hat{p}_{v_2',(z_1,z_2,u_1',v_1),W^+}(\hat{x})\hat{f}_{W^+}(\hat{x},y)\\
&\stackrel{(c)}{=}\frac{\alpha(u_1',z_1,z_2,v_1,v_2)}{\alpha(u_1',z_1,z_2,v_1,v_2')} \sum_{\hat{u}\in G_1}\hat{p}_{y_1',z_1}(\hat{u})\cdot\hat{p}_{y_2',z_2}(\hat{x}-\hat{u})e^{j2\pi\langle \hat{u},u_1'\rangle}\hat{f}_{W^+}(\hat{x},y)\\
&=\frac{P_{U_1|Z_1,Z_2,V_1,V_2}(u_1'|z_1,z_2,v_1,v_2)}{P_{U_1|Z_1,Z_2,V_1,V_2}(u_1'|z_1,z_2,v_1,v_2')}\hat{f}_{W^+}(\hat{x},y)F'(u_1')\\
&\stackrel{(d)}{=}\hat{f}_{W^+}(\hat{x},y)F'(u_1'),
\end{align*}
where (a) and (c) follow from \eqref{eqleqTek}, (b) follows from \eqref{sdnhbfvhjwsdfvhs7} and (d) follows from the fact that $I(U_1;V_2|Z_1Z_2V_1)=0$. Therefore, $F'(u_1')\neq 0$ and $F(u_1')=\hat{f}_{W^+}(\hat{x},y)F'(u_1')$.
\item If $F(u_1')=0$ then we must have $F'(u_1')=0$ (because $F'(u_1')\neq0$ would yield $F(u_1')\neq0$ in a similar way as above, which is a contradiction). Therefore, we have $F(u_1')=0=\hat{f}_{W^+}(\hat{x},y)F'(u_1')$.
\end{itemize}
We conclude that for every $u_1'\in G_1$, we have 
\begin{equation}
F(u_1')=\hat{f}_{W^+}(\hat{x},y)F'(u_1')=\sum_{\hat{u}\in G_1}\hat{f}_{W^+}(\hat{x},y)\cdot\hat{p}_{y_1',z_1}(\hat{u})\cdot\hat{p}_{y_2',z_2}(\hat{x}-\hat{u})e^{j2\pi\langle \hat{u},u_1'\rangle}.
\label{eqsBatYuIoas1}
\end{equation}

Now define $g:G_1\times G_2\rightarrow\mathbb{C}$ as follows:
\begin{equation}
g(\hat{x}',y')=\begin{cases}
\hat{f}_W(\hat{x}',y')\; &\text{if}\;(\hat{x}',y')\in\XDY(W),\\
0\;&\text{otherwise}.
\end{cases}
\end{equation}
For every $\hat{x}'\in G_1$, we have:
\begin{itemize}
\item If $\hat{p}_{y_1,z_1}(\hat{x}')\neq 0$ then $\hat{p}_{y_1',z_1}(\hat{x}')\neq 0$ (by Lemma \ref{lem3}) and $\hat{p}_{y_1,z_1}(\hat{x}')=\hat{f}_W(\hat{x}',y_1-y_1')\hat{p}_{y_1',z_1}(\hat{x}')=g(\hat{x}',y)\hat{p}_{y_1',z_1}(\hat{x}')$.
\item If $\hat{p}_{y_1,z_1}(\hat{x}')=0$ then $\hat{p}_{y_1',z_1}(\hat{x}')=0$ (by Lemma \ref{lem3}) and so $\hat{p}_{y_1,z_1}(\hat{x}')=0=g(\hat{x}',y)\hat{p}_{y_1',z_1}(\hat{x}')$.
\end{itemize}
Therefore, for every $\hat{x}'\in G_1$ we have $\hat{p}_{y_1,z_1}(\hat{x}')=g(\hat{x}',y)\hat{p}_{y_1',z_1}(\hat{x}')$. Similarly, $\hat{p}_{y_2,z_2}(\hat{x}')=g(\hat{x}',y)\hat{p}_{y_2',z_2}(\hat{x}')$ for all $\hat{x}'\in G_1$. Hence,
\begin{equation}
\label{eqsBatYuIoas2}
\begin{aligned}
F(u_1')&=\sum_{\hat{u}\in G_1}\hat{p}_{y_1,z_1}(\hat{u})\cdot\hat{p}_{y_2,z_2}(\hat{x}-\hat{u})e^{j2\pi\langle \hat{u},u_1'\rangle}\\
&=\sum_{\hat{u}\in G_1}g(\hat{u},y)\hat{p}_{y_1',z_1}(\hat{u})\cdot g(\hat{x}-\hat{u},y)\hat{p}_{y_2',z_2}(\hat{x}-\hat{u})e^{j2\pi\langle \hat{u},u_1'\rangle}.
\end{aligned}
\end{equation}
We conclude that for every $u_1'\in G_1$, we have:
\begin{equation}
\label{eqsBatYuIoas3}
\sum_{\hat{u}\in G_1}\Big[\hat{f}_{W^+}(\hat{x},y)-g(\hat{u},y)g(\hat{x}-\hat{u},y)\Big]\hat{p}_{y_1',z_1}(\hat{u})\cdot \hat{p}_{y_2',z_2}(\hat{x}-\hat{u})e^{j2\pi\langle \hat{u},u_1'\rangle}\stackrel{(a)}{=}F(u_1')-F(u_1')=0,
\end{equation}
where (a) follows from \eqref{eqsBatYuIoas1} and \eqref{eqsBatYuIoas2}. Notice that the sum in \eqref{eqsBatYuIoas3} is the inverse DFT of the function $\hat{K}:G_1\rightarrow \mathbb{C}$ defined as:
$$\hat{K}(\hat{u})=|G_1|\cdot\Big[\hat{f}_{W^+}(\hat{x},y)-g(\hat{u},y)g(\hat{x}-\hat{u},y)\Big]\hat{p}_{y_1',z_1}(\hat{u})\cdot \hat{p}_{y_2',z_2}(\hat{x}-\hat{u}).$$
Now \eqref{eqsBatYuIoas3} implies that the inverse DFT of $\hat{K}$ is zero everywhere. Therefore, $\hat{K}$ is also zero everywhere. In particular,
$$\hat{K}(\hat{x_1})=|G_1|\cdot\Big[\hat{f}_{W^+}(\hat{x},y)-g(\hat{x}_1,y)g(\hat{x}_2,y)\Big]\hat{p}_{y_1',z_1}(\hat{x}_1)\cdot \hat{p}_{y_2',z_2}(\hat{x}_2)=0.$$
But $\hat{p}_{y_1',z_1}(\hat{x}_1)\neq0$ and $\hat{p}_{y_2',z_2}(\hat{x}_2)\neq0$, so we must have $\hat{f}_{W^+}(\hat{x},y)-g(\hat{x}_1,y)g(\hat{x}_2,y)=0$. Therefore,
$$\hat{f}_{W^+}(\hat{x},y)=g(\hat{x}_1,y)g(\hat{x}_2,y)=\hat{f}_W(\hat{x}_1,y) \cdot\hat{f}_W(\hat{x}_2,y).$$
\end{enumerate}

\section{Proof of Lemma \ref{lemSufMinusPlus}}
\label{applemSufMinusPlus}

\begin{mylem}
\label{lemSufMinus}
If $W:G_1\times G_2\longrightarrow\mathcal{Z}$ is polarization compatible then $W^-$ is also polarization compatible.
\end{mylem}
\begin{proof}
Let $U_1,U_2,V_1,V_2,X_1,X_2,Y_1,Y_2,Z_1,Z_2$ be as in Remark \ref{MainRem}. Let $F:D\rightarrow \mathbb{T}$ be the pseudo-quadratic function of Definition \ref{defComp}.

Let $(\hat{u}_1,v)\in\XDY(W^-)$. There exists $z^-=(z_1,z_2)\in\mathcal{Z}^2$ such that $\hat{u}_1\in\X^{z^-}(W^-)$ and $v\in\DY^{z^-}(W^-)$. We have:
\begin{itemize}
\item Since $\hat{u}_1\in\X^{z^-}(W^-)$, there exists $v_1\in\Y^{z^-}(W^-)$ such that $\hat{p}_{v_1,z^-,W^-}(\hat{u}_1)\neq 0$. From \eqref{eqMinusLala}, we have:
\begin{equation*}
\hat{p}_{v_1,z^-,W^-}(\hat{u}_1)=\sum_{\substack{v_2\in \Y^{z_2}(W):\\v_1+v_2\in\Y^{z_1}(W)}}\frac{P_{Y_1|Z_1}(v_1+v_2|z_1)P_{Y_2|Z_2}(v_2|z_2)}{P_{V_1|Z_1,Z_2}(v_1|z_1,z_2)}\hat{p}_{v_1+v_2,z_1}(\hat{u}_1)\cdot \hat{p}_{v_2,z_2}(\hat{u}_1)^\ast.
\end{equation*}
Since $\hat{p}_{v_1,z^-,W^-}(\hat{u}_1)\neq0$, the terms in the above sum cannot all be zero. Therefore, there exists $v_2\in\Y^{z_2}(W)$ such that $v_1+v_2\in\Y^{z_1}(W)$, $\hat{p}_{v_1+v_2,z_1}(\hat{u}_1)\neq0$ and $\hat{p}_{v_2,z_2}(\hat{u}_1)\neq0$. Hence, $\hat{u}_1\in\X^{z_1}(W)$ and $\hat{u}_1\in\X^{z_2}(W)$.
\item From Lemma \ref{LemmaeqMinusLala0} we have $\Y^{z^-}(W^-)=\Y^{z_1}(W)-\Y^{z_2}(W)$ which implies that $\DY^{z^-}(W^-)=\DY^{z_1}(W)-\DY^{z_2}(W)$. Now since $v\in\DY^{z^-}(W^-)$, there exists $y_1\in \DY^{z_1}(W)$ and $y_2\in \DY^{z_2}(W)$ such that $v=y_1-y_2$.
\end{itemize}
We conclude that $$\textstyle(\hat{u}_1,y_1)\in \X^{z_1}(W)\times\DY^{z_1}(W)=\XDY^{z_1}(W)\subset \XDY(W)\subset D,$$ and $$\textstyle(\hat{u}_1,y_2)\in \X^{z_2}(W)\times\DY^{z_2}(W)= \XDY^{z_2}(W)\subset \XDY(W)\subset D.$$ Therefore, $(\hat{u}_1,v)=(\hat{u}_1,y_1-y_2)\in D$ since $D$ is a pseudo-quadratic domain. Since this is true for every $(\hat{u}_1,v)\in\XDY(W^-)$, we conclude that $\XDY(W^-)\subset D$.

Now let $(\hat{u}_1,z^-)\in\XZ(W^-)$ (where $z^-=(z_1,z_2)\in\mathcal{Z}^2$). We have shown that $\hat{u}_1\in\X^{z_1}(W)$ and $\hat{u}_1\in\X^{z_2}(W)$ and so $(\hat{u}_1,z_1)\in\XZ(W)$ and $(\hat{u}_1,z_2)\in\XZ(W)$. Fix $y_1\in\Y^{z_1}(W)$ and $y_2\in\Y^{z_2}(W)$. For every $v_1'\in\Y^{z^-}(W^-)$, we have:
\begin{align*}
&\hat{p}_{v_1',z^-,W^-}(\hat{u}_1)\\
&\;\;=\sum_{\substack{v_2'\in \Y^{z_2}(W):\\v_1'+v_2'\in\Y^{z_1}(W)}}\frac{P_{Y_1|Z_1}(v_1'+v_2'|z_1)P_{Y_2|Z_2}(v_2'|z_2)}{P_{V_1|Z_1,Z_2}(v_1'|z_1,z_2)}\hat{p}_{v_1'+v_2',z_1}(\hat{u}_1)\cdot \hat{p}_{v_2',z_2}(\hat{u}_1)^\ast\\
&\;\;\stackrel{(a)}{=}\sum_{\substack{v_2'\in \Y^{z_2}(W):\\v_1'+v_2'\in\Y^{z_1}(W)}}\frac{P_{Y_1|Z_1}(v_1'+v_2'|z_1)P_{Y_2|Z_2}(v_2'|z_2)}{P_{V_1|Z_1,Z_2}(v_1'|z_1,z_2)}\hat{p}_{y_1,z_1}(\hat{u}_1)\cdot F(\hat{u}_1,v_1'+v_2'-y_1)\cdot\frac{\hat{p}_{y_2,z_2}(\hat{u}_1)^\ast}{F(\hat{u}_1,v_2'-y_2)}\\
&\;\;\stackrel{(b)}{=}\hat{p}_{y_1,z_1}(\hat{u}_1)\cdot \hat{p}_{y_2,z_2}(\hat{u}_1)^\ast\sum_{\substack{v_2'\in \Y^{z_2}(W):\\v_1'+v_2'\in\Y^{z_1}(W)}}\frac{P_{Y_1|Z_1}(v_1'+v_2'|z_1)P_{Y_2|Z_2}(v_2'|z_2)}{P_{V_1|Z_1,Z_2}(v_1'|z_1,z_2)}F(\hat{u}_1,v_1'+v_2'-y_1-v_2'+y_2)\\
&\;\;=\hat{p}_{y_1,z_1}(\hat{u}_1)\cdot \hat{p}_{y_2,z_2}(\hat{u}_1)^\ast\cdot F(\hat{u}_1,v_1'-y_1+y_2)\sum_{\substack{v_2'\in \Y^{z_2}(W):\\v_1'+v_2'\in\Y^{z_1}(W)}}\frac{P_{Y_1|Z_1}(v_1'+v_2'|z_1)P_{Y_2|Z_2}(v_2'|z_2)}{P_{V_1|Z_1,Z_2}(v_1'|z_1,z_2)}\\
&\;\;=\hat{p}_{y_1,z_1}(\hat{u}_1)\cdot \hat{p}_{y_2,z_2}(\hat{u}_1)^\ast\cdot F(\hat{u}_1,v_1'-y_1+y_2),
\end{align*}
where (a) follows from the polarization compatibility of $W$ and from the fact that $F(\hat{u}_1,v_2'-y_2)\in\mathbb{T}$ which implies that $\displaystyle F(\hat{u}_1,v_2'-y_2)^\ast=\frac{1}{F(\hat{u}_1,v_2'-y_2)}$. (b) follows from the fact that the mapping $y\rightarrow F(\hat{u}_1,y)$ is a group homomorphism from $(H_2^{\hat{u}_1}(D),+)$ to $(\mathbb{T},\cdot)$. Therefore, for every $v_1',v_1''\in\Y^{z^-}(W^-)$, we have:
\begin{align*}
\hat{p}_{v_1',z^-,W^-}(\hat{u}_1)&=\hat{p}_{y_1,z_1}(\hat{u}_1)\cdot \hat{p}_{y_2,z_2}(\hat{u}_1)^\ast\cdot F(\hat{u}_1,v_1'-y_1+y_2)\\
&=\hat{p}_{y_1,z_1}(\hat{u}_1)\cdot \hat{p}_{y_2,z_2}(\hat{u}_1)^\ast\cdot F(\hat{u}_1,v_1''-y_1+y_2+v_1'-v_1'')\\
&=\hat{p}_{y_1,z_1}(\hat{u}_1)\cdot \hat{p}_{y_2,z_2}(\hat{u}_1)^\ast\cdot F(\hat{u}_1,v_1''-y_1+y_2)\cdot F(\hat{u}_1,v_1'-v_1'')\\
&=\hat{p}_{v_1'',z^-,W^-}(\hat{u}_1)\cdot F(\hat{u}_1,v_1'-v_1'').
\end{align*}
Hence, $\hat{p}_{v_1',z^-,W^-}(\hat{u}_1)=F(\hat{u}_1,v_1'-v_1'')\cdot\hat{p}_{v_1'',z^-,W^-}(\hat{u}_1)$. We conclude that $W^-$ is polarization compatible.
\end{proof}

\begin{mylem}
\label{lemSufPlus}
If $W:G_1\times G_2\longrightarrow\mathcal{Z}$ is polarization compatible then $W^+$ is also polarization compatible.
\end{mylem}
\begin{proof}
Let $U_1,U_2,V_1,V_2,X_1,X_2,Y_1,Y_2,Z_1,Z_2$ be as in Remark \ref{MainRem}. Let $F:D\rightarrow \mathbb{T}$ be the pseudo-quadratic function of Definition \ref{defComp}.

Let $(\hat{u}_2,v)\in \XDY(W^+)$. There exists $z^+=(z_1,z_2,u_1,v_1)\in\mathcal{Z}^+$ such that $\hat{u}_2\in\X^{z^+}(W^+)$ and $v\in\DY^{z^+}(W^+)$. We have:
\begin{itemize}
\item Since $\hat{u}_2\in\X^{z^+}(W^+)$, there exists $v_2\in\Y^{z^+}(W^+)$ such that $\hat{p}_{v_2,z^+}(\hat{u}_2)\neq0$. From \eqref{eqleqTek} we have
\begin{equation*}
\hat{p}_{v_2,z^+,W^+}(\hat{u}_2)=\sum_{\hat{u}_2'\in G_1}\frac{\hat{p}_{v_1+v_2,z_1}(\hat{u}_2')\cdot\hat{p}_{v_2,z_2}(\hat{u}_2-\hat{u}_2')}{|G_1|\alpha(u_1,z_1,z_2,v_1,v_2)}e^{j2\pi\langle \hat{u}_2',u_1 \rangle}.
\end{equation*}
Since $\hat{p}_{v_2,z^+,W^+}(\hat{u}_2)\neq0$, there must exist $\hat{u}_2'\in G_1$ such that $\hat{p}_{v_1+v_2,z_1}(\hat{u}_2')\neq0$ and $\hat{p}_{v_2,z_2}(\hat{u}_2-\hat{u}_2')\neq0$. Therefore, $\hat{u}_2'\in\X^{z_1}(W)$ and $(\hat{u}_2-\hat{u}_2')\in\X^{z_2}(W)$.
\item Since $v\in\DY^{z^+}(W^+)$, there exist $v_2',v_2''\in \Y^{z^+}(W^+)$ such that $v=v_2'-v_2''$. Now Lemma \ref{lemPlushaha} implies that $v_1+v_2'\in\Y^{z_1}(W)$, $v_2'\in\Y^{z_2}(W)$, $v_1+v_2''\in\Y^{z_1}(W)$ and $v_2''\in\Y^{z_2}(W)$. Therefore,
$v=(v_1+v_2')-(v_1+v_2'')\in\DY^{z_1}(W)$ and $v=v_2'-v_2''\in\DY^{z_2}(W)$.
\end{itemize}
We conclude that $$\textstyle(\hat{u}_2',v)\in \X^{z_1}(W)\times\DY^{z_1}(W)=\XDY^{z_1}(W)\subset \XDY(W)\subset D$$
and
$$\textstyle(\hat{u}_2-\hat{u}_2',v)\in \X^{z_2}(W)\times\DY^{z_2}(W)=\XDY^{z_2}(W)\subset \XDY(W)\subset D.$$
Now since $D$ is a pseudo-quadratic domain, we have $(\hat{u}_2,v)=\big(\hat{u}_2'+(\hat{u}_2-\hat{u}_2'),v\big)\in D$. We conclude that $\XDY(W^+)\subset D$.

Now let $(\hat{u}_2,z^+)\in \XZ(W^+)$, where $z^+=(z_1,z_2,u_1,v_1)\in\mathcal{Z}^+$. For every $v_2',v_2''\in\Y^{z^+}(W^+)$, we have $v_1+v_2'\in\Y^{z_1}(W)$, $v_2'\in\Y^{z_2}(W)$, $v_1+v_2''\in\Y^{z_1}(W)$ and $v_2''\in\Y^{z_2}(W)$ from Lemma \ref{lemPlushaha}. Therefore,
\begin{align*}
\hat{p}_{v_2',z^+,W^+}(\hat{u}_2)&=\sum_{\hat{u}_2'\in G_1}\frac{\hat{p}_{v_1+v_2',z_1}(\hat{u}_2')\cdot\hat{p}_{v_2',z_2}(\hat{u}_2-\hat{u}_2')}{|G_1|\alpha(u_1,z_1,z_2,v_1,v_2')}e^{j2\pi\langle \hat{u}_2',u_1 \rangle}\\
&\stackrel{(a)}{=}\sum_{\substack{\hat{u}_2'\in G_1:\\\hat{u}_2'\in\X^{z_1}(W),\\\hat{u}_2-\hat{u}_2'\in\X^{z_2}(W)}}\frac{\hat{p}_{v_1+v_2',z_1}(\hat{u}_2')\cdot\hat{p}_{v_2',z_2}(\hat{u}_2-\hat{u}_2')}{|G_1|\alpha(u_1,z_1,z_2,v_1,v_2')}e^{j2\pi\langle \hat{u}_2',u_1 \rangle}\\
&\stackrel{(b)}{=}\sum_{\substack{\hat{u}_2'\in G_1:\\\hat{u}_2'\in\X^{z_1}(W),\\\hat{u}_2-\hat{u}_2'\in\X^{z_2}(W)}}\frac{\hat{p}_{v_1+v_2'',z_1}(\hat{u}_2')F(\hat{u}_2',v_2'-v_2'')\cdot\hat{p}_{v_2'',z_2}(\hat{u}_2-\hat{u}_2')F(\hat{u}_2-\hat{u}_2',v_2'-v_2'')}{|G_1|\cdot P_{U_1|Z_1,Z_2,V_1,V_2}(u_1|z_1,z_2,v_1,v_2')}e^{j2\pi\langle \hat{u}_2',u_1 \rangle}\\
&\stackrel{(c)}{=}\sum_{\hat{u}_2'\in G_1}\frac{\hat{p}_{v_1+v_2'',z_1}(\hat{u}_2')\cdot\hat{p}_{v_2'',z_2}(\hat{u}_2-\hat{u}_2')}{|G_1|\cdot P_{U_1|Z_1,Z_2,V_1,V_2}(u_1|z_1,z_2,v_1,v_2'')}F(\hat{u}_2'+\hat{u}_2-\hat{u}_2',v_2'-v_2'')\cdot e^{j2\pi\langle \hat{u}_2',u_1 \rangle}\\
&=F(\hat{u}_2,v_2'-v_2'')\sum_{\hat{u}_2'\in G_1}\frac{\hat{p}_{v_1+v_2'',z_1}(\hat{u}_2')\cdot\hat{p}_{v_2'',z_2}(\hat{u}_2-\hat{u}_2')}{|G_1|\cdot P_{U_1|Z_1,Z_2,V_1,V_2}(u_1|z_1,z_2,v_1,v_2'')}e^{j2\pi\langle \hat{u}_2',u_1 \rangle}\\
&=F(\hat{u}_2,v_2'-v_2'')\hat{p}_{v_2'',z^+,W^+}(\hat{u}_2),
\end{align*}

where (a) follows from the fact that $\hat{p}_{v_1+v_2',z_1}(\hat{u}_2')=0$ if $\hat{u}_2'\notin\X^{z_1}(W)$, and $\hat{p}_{v_2',z_2}(\hat{u}_2-\hat{u}_2')=0$ if $(\hat{u}_2-\hat{u}_2')\notin\X^{z_2}(W)$. (b) follows from the fact that $W$ is polarization compatible. (c) follows from the fact that $F$ is pseudo-quadratic and the fact that $U_1$ is conditionally independent of $V_2$ given $(Z_1,Z_2,V_1)$ (since the polarization compatibility of $W$ implies that $I_1$ is preserved for $W$ by Lemma \ref{lemSuf1}, which implies that $I(U_1;V_2|Z_1Z_2V_1)=0$). Therefore, for every $v_2',v_2''\in\Y^{z^+}(W^+)$, we have $$\hat{p}_{v_2',z^+,W^+}(\hat{u}_2)=F(\hat{u}_2,v_2'-v_2'')\cdot\hat{p}_{v_2'',z^+,W^+}(\hat{u}_2).$$ We conclude that $W^+$ is polarization compatible.
\end{proof}

\vspace*{3mm}

Lemma \ref{lemSufMinusPlus} follows from Lemmas \ref{lemSufMinus} and \ref{lemSufPlus}.

\bibliographystyle{IEEEtran}
\bibliography{bibliofile}
\end{document}